\documentclass[11pt,a4paper]{article}
\usepackage{hyperref}
\usepackage{amsmath,amsfonts,amssymb,amsthm,mathtools}
\usepackage{latexsym,graphicx}
\usepackage{dsfont}
\usepackage{comment}
\usepackage{amscd}
\usepackage[dvipsnames]{xcolor}
\usepackage{extarrows} 
	\numberwithin{equation}{section} 
\usepackage{yfonts}
\usepackage{graphicx}
\usepackage{appendix}

\allowdisplaybreaks

\def\Xint#1{\mathchoice
{\XXint\displaystyle\textstyle{#1}}%
{\XXint\textstyle\scriptstyle{#1}}%
{\XXint\scriptstyle\scriptscriptstyle{#1}}%
{\XXint\scriptscriptstyle\scriptscriptstyle{#1}}%
\!\int}
\def\XXint#1#2#3{{\setbox0=\hbox{$#1{#2#3}{\int}$}
\vcenter{\hbox{$#2#3$}}\kern-.5\wd0}}

\def\pvint{\Xint-}
\newcommand{\ee}{{\rm e}}
\newcommand{\ii}{{\rm i}}

\newcommand{\R}{{\mathbb R}}
\newcommand{\C}{{\mathbb C}}
\newcommand{\Z}{{\mathbb Z}}

\newtheorem{theorem}{Theorem}[section]
\newtheorem{proposition}{Proposition}[section]
\newtheorem{lemma}{Lemma}[section]
\newtheorem{remark}{Remark}[section]

\newcommand{\cT}{\mathcal{T}}

\newcommand{\ftwo}{\varkappa}
\newcommand{\cc}{\gamma_0}

\newcommand{\re}{\mathrm{Re} }
\newcommand{\im}{\mathrm{Im} }

\title{\Large{A focusing-defocusing intermediate nonlinear Schr\"{o}dinger system}}

\author{Bjorn K. Berntson \\ bbernts@kth.se \and Alexander Fagerlund \\ afage@kth.se}

\date{Department of Physics, KTH Royal Institute of Technology, Stockholm, Sweden \\
\vspace{1em}
\today}

\begin{document}

\bigskip

{\let\newpage\relax\maketitle}

\maketitle

\begin{abstract}

We introduce and study a system of coupled nonlocal nonlinear Schr\"{o}dinger equations that interpolates between the mixed, focusing-defocusing Manakov system on one hand and a limiting case of the intermediate nonlinear Schr\"{o}dinger equation on the other. We show that this new system, which we call the intermediate mixed Manakov (IMM) system, admits multi-soliton solutions governed by a complexification of the hyperbolic Calogero-Moser (CM) system. Furthermore, we introduce a spatially periodic version of the IMM system, for which our result is a class of exact solutions governed by a complexified elliptic CM system.
\end{abstract} 

\noindent
{\small{\sc AMS Subject Classification (2020)}: 33E05, 35C08, 35Q51, 35Q55, 35Q70}
 
\noindent
{\small{\sc Keywords}: {integrable system, coupled nonlinear Schr\"{o}dinger equations, solitons, elliptic functions}

\section{Introduction} 

The nonlinear Schr\"{o}dinger (NLS) equation is a fundamental model both for wave propagation in weakly nonlinear, dispersive media \cite{benney1967} and in the theory of integrable systems \cite{ablowitz2003discrete}. The ability of the NLS equation to model a variety of nonlinear physics together with its amenability to exact analytic methods has inspired the development of various integrability-preserving generalizations and relative equations. Prominent among these is the Manakov system \cite{manakov1974,zakharov1982},
\begin{equation}\label{eq:manakov}
\begin{split}
\ii u_t=&\; u_{xx}+ u(\sigma_1 |u|^2+\sigma_2|v|^2), \\
\ii v_t=&\; v_{xx}+ v(\sigma_1 |u|^2+\sigma_2 |v|^2)
\end{split}\qquad  (\sigma_1,\sigma_2=\pm 1),
\end{equation}
a two-component variant of the NLS equation with applications to optics \cite{berkhoer1970}, water waves \cite{baronio2012}, and Bose-Einstein condensates \cite{bludov2010}. The Manakov system \eqref{eq:manakov} comes in three cases, up to equivalence.\footnote{The cases $\sigma_1=-\sigma_2=+1$ and $\sigma_1=-\sigma_2=-1$ of \eqref{eq:manakov} are equivalent via $u\leftrightarrow v$.} The focusing, $\sigma_1=\sigma_2=+1$, and defocusing, $\sigma_1=\sigma_2=-1$, cases of \eqref{eq:manakov} are generalizations of the corresponding cases of the NLS equation, which are recovered via the reduction $v=u$. The mixed, focusing-defocusing Manakov system, \eqref{eq:manakov} with $\sigma_1\sigma_2=-1$, does not reduce to a one-component NLS equation, but is an interesting integrable system in its own right, having recently been studied from the perspectives of boundary value problems \cite{tian2016,tian2017} and soliton phenomena \cite{kanna2006,ohta2011general,feng2014}. In this paper, we connect the mixed Manakov system to nonlocal (integro-differential) NLS systems by introducing a nonlocal deformation of the former. 

The \textit{intermediate mixed Manakov} (IMM) system reads
\begin{equation}\label{eq:IMM}
\begin{split}
\ii u_t=&\; u_{xx}+ u(\ii + T)(|u|^2)_x- u\tilde{T}(|v|^2)_x     , \\
\ii v_t=&\; v_{xx}+ v(\ii - T)(|v|^2)_x + v\tilde{T}(|u|^2)_x,
\end{split}
\end{equation}
with the integral operators
\begin{equation} 
\label{eq:TT}
\begin{split} 
(Tf)(x) \coloneqq &\; \frac1{2\delta}\pvint_{\R}\coth\left(\frac{\pi}{2\delta}(x'-x)\right)f(x')\,\mathrm{d}x', \\
(\tilde{T} f)(x) \coloneqq &\; \frac1{2\delta}\int_{\R}\tanh\left(\frac{\pi}{2\delta}(x'-x)\right)f(x')\,\mathrm{d}x',
\end{split} 
\end{equation} 
where $\delta>0$ is an arbitrary parameter and the dashed integral indicates a principal value prescription at $x'=x$. The intermediacy of \eqref{eq:IMM} corresponds to the fact, elaborated in Section~\ref{subsec:basic}, that it interpolates between the mixed Manakov system, which is obtained in the limit $\delta\downarrow 0$, and the following system of uncoupled Hilbert NLS (HNLS) equations \cite{matsuno2000,gerard2022},\begin{equation}\label{eq:HNLS}
\begin{split}
\ii u_t=&\;  u_{xx}+ u(\ii+ H)(|u|^2)_x, \\
\ii v_t= &\; v_{xx}+ v(\ii- H)(|v|^2)_x 
\end{split}
\end{equation}
with $H$ the Hilbert transform,
\begin{equation}\label{eq:H}
(Hf)(x)\coloneqq \frac{1}{\pi} \pvint_{\R} \frac{f(x')}{x'-x}\,\mathrm{d}x',
\end{equation}
which is obtained in the $\delta\to +\infty$ limit. In addition to generalizing the mixed Manakov system, the IMM system falls neatly into two established classes of integrable systems, as we now describe.
\begin{enumerate}
\item The study of integro-differential NLS equations was initiated by Pelinovsky in \cite{pelinovsky1995}, where the intermediate NLS (INLS) equation
\begin{equation}\label{eq:INLS}
	\ii u_t=u_{xx}+u(\ii-\sigma T)(|u|^2)_x \quad (\sigma=\pm 1)
\end{equation}
was derived as a description of envelope waves in the intermediate long wave equation. The INLS equation generalizes the standard NLS equation, which is recovered in the $\delta\downarrow 0$ limit. An inverse scattering transform for the defocusing $(\sigma=-1)$ INLS equation has been developed \cite{pelinovsky1995spectral} and multi-soliton solutions have been found in both the defocusing and focusing ($\sigma=+1)$ cases \cite{pelinovsky1995,matsuno2001,tutiya2009}. Similar results for the HNLS equation \eqref{eq:HNLS}, which is obtained from the INLS equation in the $\delta\to +\infty$ limit, have been established by Matsuno \cite{matsuno2000,matsuno2002exactly,matsuno2004}. Moreover, and of particular relevance to us, Matsuno established that soliton and certain spatially periodic solutions of the defocusing HNLS equation (i.e., the first equation in \eqref{eq:HNLS}) are governed by the rational and trigonometric Calogero-Moser (CM) systems, respectively, subject to certain constraints on their initial conditions \cite{matsuno2002}. G\'{e}rard and Lenzmann have recently established several rigorous results for the focusing HNLS equation (i.e., the second equation in \eqref{eq:HNLS}) including a novel Lax pair structure and global-in-time multi-soliton solutions \cite{gerard2022}. The IMM system is the first-studied two-component system in this class. 

\item Several exactly-solvable systems involving both the $T$ and $\tilde{T}$ operators \eqref{eq:TT} have recently been introduced \cite{berntson2020,berntson2022,berntsonlangmannlenells2022}; the IMM system provides a further example in this class. A prominent feature of known such examples is the existence of families of solutions governed by CM systems. We will show that the IMM system likewise has such solutions. However, in contrast to known examples, which only have meaningful limits as $\delta\to +\infty$, the IMM system is a genuine intermediate system in the sense described above. 
\end{enumerate}

Our results on the IMM system are motivated by and use tools from the study of the classes of systems discussed above. Most importantly, we view our results as hyperbolic and elliptic generalizations of Matsuno's work \cite{matsuno2002} on solving an integro-differential NLS equation (the HNLS equation \eqref{eq:HNLS}) using rational and trigonometric CM systems. In the remainder of this section, we describe these results, establish certain basic properties of the IMM system, outline our plan for the paper, and introduce the notation we use in the main text. 

\subsection{Summary of results}

We establish a precise connection between the IMM system and the hyperbolic and elliptic cases of the CM system, defined for $N\in \Z_{\geq 1}$ to be the system of ordinary differential equations (ODEs),
\begin{equation}\label{eq:CM}
\ddot{a}_j=-4 \sum_{k\neq j}^N V'(a_j-a_k) 	\quad (j=1,\ldots,N)
\end{equation}
with
\begin{equation}\label{eq:V}
V(z)\coloneqq \begin{cases}
1/z^2 & \text{(I: rational case)} \\ 
(\pi/2\ell)^2/\sin^2(\pi z/2\ell) & \text{(II: trigonometric case)} \\
(\pi/2\delta)^2/\sinh^2(\pi z/2\delta) & \text{(III: hyperbolic case)}\\
\wp_2(z;\ell,\ii\delta) & \text{(IV: elliptic case)},  	
\end{cases}	
\end{equation}
where $\ell,\delta>0$ are free parameters and in case IV, $\wp_2(z;\ell,\ii\delta)$ is equal to the Weierstrass elliptic function with half-periods $(\ell,\ii\delta)$ up to an additive constant (see \eqref{eq:wp2} for the precise definition, but note that the value of the constant is irrelevant for \eqref{eq:CM}). These systems are exactly solvable \cite{olshanetsky1983,krichever1980} and, as we describe below, give rise to exact solutions to the IMM system. As mentioned above, our results extend a known correspondence between the HNLS equation and the rational and trigonometric cases of the CM system \cite{matsuno2002,gerard2022}. 

More specifically, we construct solutions of the IMM system as the values on the real line of certain meromorphic functions (one for $u$ and one for $v$), following an idea due to Kruskal \cite{kruskal1974} (see also \cite{thickstun1976,airault1977,choodnovsky1977}) and applied to the HNLS system by Matsuno \cite{matsuno2002} (see also \cite{gerard2022}). We make ans\"{a}tze for these functions involving (dynamical) parameters that determine the (simple) poles and the corresponding residues and an additive term, which we refer to as the background.

Our first result, whose precise statement is given in Theorem~\ref{thm:solitons}, provides a recipe to construct $N$-soliton solutions of the IMM system as the real-line values of $2\ii\delta$-periodic, meromorphic functions, each with a constant background and $N$ dynamical simple poles and corresponding residues. The poles of the meromorphic functions are determined by certain solutions $\{a_j\}_{j=1}^N$ of the complexified hyperbolic CM system. The corresponding residues $\{c_j\}_{j=1}^N$ are obtained as certain solutions of a linear system of ODEs with coefficients depending on $\{a_j\}_{j=1}^N$.

Our second result is an adaptation of the first result that applies to the IMM system with $2\ell$-periodic boundary conditions; the precise statement is given in Theorem~\ref{thm:ellipticsolitons}. We construct $2\ell$-periodic solutions of the IMM system as the real-line values of $2\ell$- and $2\ii\delta$-periodic, meromorphic (i.e., elliptic) functions, each having a dynamical background and $N$ dynamical simple poles and corresponding residues. The parameters $\{a_j\}_{j=1}^N$ are certain solutions of the complexified elliptic CM system and determine the poles of the elliptic solutions to the IMM system. Similarly as in the first result, the parameters $\{c_j\}_{j=1}^N$ determine the residues corresponding to these poles and solve a linear system of ODEs with coefficients depending on $\{a_j\}_{j=1}^N$. The dynamical nature of the background is a new feature in the periodic setting; the dynamics is determined by the value of a quantity $\lambda$, which solves an ODE of the form $\dot{\lambda}=F\big(\{a_j\}_{j=1}^N,\{c_j\}_{j=1}^N\big)$.

Our results are supplemented by examples with corresponding visualizations. The key step in constructing such examples is to solve certain nonlinear constraints that the initial values of the time-dependent parameters must satisfy. In the case of our first result, we develop methods to do this when $N=1,2,3$; this process gives rise to one-, two-, and three-soliton solutions of the IMM system, which strongly suggests that the IMM system is integrable (and that it admits $N$-soliton solutions for $N>3$) \cite{ramani1981,hietarinta1987}. We also develop methods to solve the analogous constraints in the $N=2,3$ cases of our second result (we exclude the $N=1$ case from our second result as it would give only trivial (i.e., constant) solutions; see Section~\ref{sec:periodic} for details).

\subsection{Basic properties of the IMM system}\label{subsec:basic}

We collect some useful properties of the IMM system: limits to known systems, symmetries that we apply in the main text, and a two-vector notation used in the proof of our results. 

\paragraph{Limits.} The nonlocal operators in the IMM system have simple representations as Fourier multipliers, from which several basic properties (such as limits) of the IMM system follow.  With the convention $\hat{f}(k)\coloneqq \int_{\R} f(x)\ee^{-\ii kx}\,\mathrm{d}x$, the Fourier transforms of the operators $T$ and $\tilde{T}$ in \eqref{eq:TT} are \cite{berntson2020}
\begin{equation}
\begin{split}
(\widehat{Tf})(k)=&\; \ii\,\coth(k\delta)\hat{f}(k), \\
(\widehat{\tilde{T}f})(k)=&\; \frac{\ii}{\sinh(k\delta)}\hat{f}(k)  
\end{split}\qquad (k\in \R\setminus\{0\}). 
\end{equation}
Thus, in the limit $\delta\rightarrow+\infty$,
\begin{equation}\label{eq:deepwaterT}
\begin{split}
(\widehat{Tf})(k)\rightarrow&\; \ii\,\mathrm{sgn}(k)\hat{f}(k),\\
(\widehat{\tilde{T}f})(k)\rightarrow &\; 0 
\end{split}\qquad (k\in \R\setminus\{0\}).
\end{equation}  
We recognize the Fourier multiplier $\ii\,\mathrm{sgn}(k)$ as that of the Fourier space representation of the Hilbert transform \eqref{eq:H}. Hence, taking the $\delta\to +\infty$ limit of \eqref{eq:IMM} and using \eqref{eq:deepwaterT}, we obtain the HNLS equation \eqref{eq:HNLS}. 

In the limit $\delta\downarrow 0$, known asymptotic expansions of the $T$ and $\tilde{T}$ operators \cite{berntsonlangmann2020}  can be used to show that the IMM equation \eqref{eq:IMM} reduces to the mixed Manakov system \eqref{eq:manakov} with $\sigma_1=-\sigma_2=+1$. Details are given in Appendix~\ref{app:local}.     

\paragraph{Symmetries.}
The IMM system, like the Manakov system, possesses the $\mathrm{U}(1)\times\mathrm{U}(1)$ symmetry,
\begin{equation}\label{eq:U1timesU1}
(u,v)\to (\ee^{\ii\theta_1}u, \ee^{\ii\theta_2}v) \qquad (\theta_1,\theta_2\in\mathbb{R})  
\end{equation}
and the Galilean symmetry,
\begin{equation}\label{eq:galilean}
\big(u(x,t),v(x,t)\big)\to \ee^{-\ii \eta x+\ii \eta^2 t}\big(u(x-2\eta t,t),v(x-2\eta t,t)\big) \qquad (\eta\in \R).
\end{equation}

It is interesting to note that, in contrast to known systems involving $T$ and $\tilde{T}$ operators, the IMM system does not admit a natural discrete symmetry interchanging the two component equations (which can be interpreted as non-chirality \cite{berntson2020,berntson2022,berntsonlangmannlenells2022}). To be more specific, there is no combination of variable interchange ($u\leftrightarrow v$), parity inversion ($x\to -x)$, time reversal ($t\to -t$), and complex conjugation that interchanges the component equations in \eqref{eq:IMM}.
  
\paragraph{Two-vector notation.}  
  
The IMM system can be written as a single equation using the following notation developed within similar contexts \cite{berntson2020,berntson2022}.  
  
  Given $\C$-valued functions $F_j,G_j$, $j=1,2$, we define the product 
\begin{equation}\label{eq:circ}
\left(\begin{array}{c} F_1 \\ F_2  \end{array}\right)\circ \left(\begin{array}{c} G_1 \\ G_2    \end{array}\right)\coloneqq \left(\begin{array}{c} F_1G_1 \\ -F_2 G_2    \end{array}\right)
\end{equation}
and the linear operators
\begin{equation}\label{eq:cT}
\cT: \left(\begin{array}{c} F_1 \\ F_2  \end{array}\right)\mapsto \left(\begin{array}{cc} T & \tilde{T} \\ -\tilde{T} & -T    \end{array}\right)\left(\begin{array}{c} F_1 \\ F_2  \end{array}\right) \coloneqq \left(\begin{array}{c} TF_1+\tilde{T}F_2 \\ -\tilde{T} F_1-TF_2  \end{array}\right),
\end{equation}
with $T$ and $\tilde{T}$ as in \eqref{eq:TT}. With this, the IMM equation \eqref{eq:IMM} can be written as
\begin{equation}\label{eq:IMM2}
\ii U_t=U_{xx}+ U\circ (\ii+\cT)(U\circ U^*)_x,\qquad U\coloneqq \left(\begin{array}{c} u \\ v \end{array}\right).
\end{equation}
  In this way, the IMM system is a natural two-component variant of the INLS equation \eqref{eq:INLS}. 
  
\subsection{Plan of the paper}

 In Section~\ref{sec:solitons}, an ansatz is used to construct soliton solutions of the IMM system controlled by the hyperbolic CM system. The periodic version of the IMM system is introduced in Section~\ref{sec:periodic}, where we also extend the results of Section~\ref{sec:solitons} to the periodic setting. Properties of the special functions we use are collected in Appendix~\ref{app:functional}. Appendix~\ref{app:proofs} is devoted to a proof of the main result of Section~\ref{sec:periodic}, which (as we explain in Section~\ref{subsubsec:limit}) essentially contains the proof of the main result of Section~\ref{sec:solitons} as a special case. In Appendix~\ref{app:local}, we provide details on the $\delta\downarrow 0$ limit of the IMM system discussed in Section~\ref{subsec:basic}.
Appendix~\ref{app:amplitudes} contains a derivation of useful expressions  for the (squared) amplitudes of our solutions.
\subsection{Notation}

We use the shorthand notation $\sum_{k\neq j}^{N}$ for sums $\sum_{k=1,k\neq j}^N$, etc. A dot above a variable indicates differentiation with respect to time while a prime indicates differentiation with respect to the argument of a function. Complex conjugation is denoted by $*$. The two-vector notation introduced in Section~\ref{subsec:basic} is used in Appendix~\ref{app:proofs} and at selected points in the main text. 

\section{Solitons}\label{sec:solitons}

We will construct multi-soliton solutions of the IMM system by making an ansatz with time-dependent complex poles and residues and showing that the poles evolve according to a complexified version of the hyperbolic CM system while the residues solve a linear system of ODEs. To be more concrete, we first introduce the following special function that will play a key role in our analysis,
\begin{equation}\label{eq:alpha}
\alpha(z)\coloneqq \frac{\pi}{2\delta}\coth\bigg(\frac{\pi }{2\delta}z\bigg)	;
\end{equation}
note that $V(z)$ in Case III of \eqref{eq:V} is equal to $-\alpha'(z)$. 

Our ansatz for the soliton solutions of the IMM system is
\begin{equation}\label{eq:ansatz}
\left(\begin{array}{c} u(x,t) \\ v(x,t) \end{array}\right) = \lambda \left(\begin{array}{c} 1 \\ -1 \end{array}\right)  + \ii\sum_{j=1}^N c_j(t) \left(\begin{array}{c} \alpha(x-a_j(t)-\ii\delta/2) \\ -\alpha(x-a_j(t)+\ii\delta/2)   \end{array}\right),
\end{equation}
where $\lambda$ is a constant, which, without loss of generality, may be chosen to be real by the $\mathrm{U}(1)\times\mathrm{U}(1)$ invariance of the IMM system \eqref{eq:U1timesU1}, and $\{a_j,c_j\}_{j=1}^N$ are complex-valued functions of $t$.

The key necessary condition for the ansatz \eqref{eq:ansatz} to be consistent is that $\{a_j\}_{j=1}^N$ satisfy the hyperbolic CM system, \eqref{eq:CM} with $V(z)$ in Case III of \eqref{eq:V}.  The parameters $\{c_j\}_{j=1}^N$ must satisfy the following system of first-order linear ODEs,
\begin{equation}\label{eq:cjdot}
\dot{c}_j= 2\ii \sum_{k\neq j}^N (c_j-c_k)V(a_j-a_k) \quad (j=1,\ldots,N).
\end{equation}

The precise statement of our result is given as Theorem~\ref{thm:solitons} in Section~\ref{subsec:result}. One-, two-, and three-soliton solutions, corresponding to the cases $N=1$, $2$,  and $3$ of Theorem~\ref{thm:solitons} are given special attention in Section~\ref{subsec:examples}. 

\subsection{Result}\label{subsec:result}

Our result is stated and followed by several remarks. 

\begin{theorem}\label{thm:solitons}
For $N\in \Z_{\geq 1}$ and $\lambda\in \R$, let $\{a_j,c_j\}_{j=1}^N$ be a solution of the system of ODEs consisting of \eqref{eq:CM} and \eqref{eq:cjdot} on an interval $[0,\tau)$ for some $\tau \in (0,\infty)\cup\{\infty\}$ and with initial conditions that satisfy
\begin{equation}\label{eq:ajdot}
c_j\dot{a}_j=2\lambda+2\ii\sum_{k\neq j}^N c_k \alpha(a_j-a_k) \quad (j=1,\ldots,N)
\end{equation}
and
\begin{equation}\label{eq:constraint}
c_j\Bigg(\lambda-\ii\sum_{k=1}^N c_k^*\alpha(a_j-a_k^*+\ii\delta)\Bigg)+1=0 \quad (j=1,\ldots,N)
\end{equation}
at $t=0$. Moreover, suppose that the conditions
\begin{equation}\label{eq:imaj}
-\frac{3\delta}{2}< \im(a_j) < -\frac{\delta}{2} \quad (j=1,\ldots,N)
\end{equation}
and
\begin{equation}\label{eq:ajak}
a_j\neq a_k \quad (1\leq j<k\leq N)
\end{equation}
are satisfied for $t\in [0,\tau)$. Then, \eqref{eq:ansatz} solves the IMM system \eqref{eq:IMM}--\eqref{eq:TT} on $[0,\tau)$. 
\end{theorem}

\subsubsection{Remarks on Theorem~\ref{thm:solitons}}\label{subsubsec:remarks}

\begin{enumerate}
\item We omit the proof of Theorem~\ref{thm:solitons} because, as will be elaborated in Section~\ref{subsubsec:remarksperiodic}, it is essentially a special case of Theorem~\ref{thm:ellipticsolitons}, which is stated in Section~\ref{subsec:resultelliptic} and proven in detail in Appendix~\ref{app:proofs}. 
\item Our solutions \eqref{eq:ansatz} obey the boundary conditions
\begin{equation}\label{eq:bc}
\lim_{x\to\pm\infty} u(x,t)=-\lim_{x\to\pm\infty} v(x,t)=\lambda\pm \ii\frac{\pi}{2\delta}\sum_{j=1}^N c_j(t),
\end{equation}
which follow from
\begin{equation}
\lim_{x\to\pm\infty} \alpha(x+\ii y)=\frac{\pi}{2\delta} \lim_{x\to\pm\infty} \coth\bigg(\frac{\pi}{2\delta}(x+\ii y)\bigg)=\pm \frac{\pi}{2\delta}. 
\end{equation}
The boundary conditions \eqref{eq:bc} are time-independent. To see this, observe that \eqref{eq:cjdot} implies that $\sum_{j=1}^N c_j$ is conserved in time (see Lemma~\ref{lem:cjconstraint} in Appendix~\ref{app:thmproof} for details of this calculation). It can further be seen, using \eqref{eq:constraint}, that $\sum_{j=1}^N \im(c_j)=0$ (see Sections~\ref{subsubsec:twosol}--\ref{subsubsec:threesol} for details in the cases $N=2,3$). Consequently, the second term determining the boundary conditions in \eqref{eq:bc} is purely imaginary.

\item The ansatz \eqref{eq:ansatz} (or its complex conjugate appearing implicitly in \eqref{eq:IMM}) can be written in terms of the functions
\begin{equation}\label{eq:Apm}
A_{\pm}(z)\coloneqq\left(\begin{array}{c}\alpha(z\mp\ii\delta/2) \\ -\alpha(z\pm\ii\delta/2)\end{array}\right). 
\end{equation}
using the two-vector notation introduced in Section~\ref{subsec:basic}. To establish our result, a crucial feature of the functions \eqref{eq:Apm} is that their derivatives $A_{\pm}'(x-a_j)$ are eigenfunctions of the operator $\cT$ in \eqref{eq:cT} \cite{berntson2020},
\begin{equation}
(\cT A_{\pm}'(\cdot-a_j))(x)=\pm \ii A_{\pm}'(x-a_j)	
\end{equation} 
when \eqref{eq:imaj} holds. 
\item The ansatz \eqref{eq:ansatz} satisfies the IMM system provided that the first-order system \eqref{eq:cjdot}--\eqref{eq:ajdot} is satisfied and the conditions \eqref{eq:imaj}--\eqref{eq:ajak} hold. This first-order system is equivalent to the system of equations \eqref{eq:CM} and \eqref{eq:cjdot} when equipped with compatible initial conditions satisfying \eqref{eq:constraint} and when \eqref{eq:ajak} holds; see Proposition~\ref{prop:CM} in Appendix~\ref{app:thmproof} for the precise statement. In this way, the first-order system resembles the well-known B\"{a}cklund transformation for the CM system \cite{wojciechowski1982}. 
\item The condition \eqref{eq:ajak} excludes the possibility of pole collisions on the interval $[0,\tau)$. Such collisions cannot occur in (repulsive) CM systems \eqref{eq:CM}--\eqref{eq:V} with real-valued $\{a_j\}_{j=1}^N$, but are possible in the complexified CM systems we use; see \cite{wilson1998,gerard2022} for details in Case I of \eqref{eq:CM}--\eqref{eq:V}.  
\end{enumerate}

\subsection{Examples of solutions}\label{subsec:examples}

We present examples of solutions coming from the $N=1,2,3$ cases of Theorem~\ref{thm:solitons}. Doing so involves solving the nonlinear constraints \eqref{eq:constraint} at $t=0$ and using the equations of motion \eqref{eq:CM}, 
\eqref{eq:cjdot} to propagate this initial date forward in time. In the case $N=1$ we obtain a fully explicit solution, while in the cases $N=2,3$ we give 
	numerical solutions\footnote{We note that the hyperbolic CM system can be solved exactly by linear algebra methods \cite{olshanetsky1983}. Moreover, a method developed by Matsuno \cite{matsuno2002} to exactly solve a linear system of ODEs analogous to \eqref{eq:cjdot} could be adapted for our purposes. However, for efficiency and due to the fact that exactly solving the elliptic CM system \cite{krichever1980} is a more challenging procedure, we focus on generating numerical examples in this paper.} arising from initial data that exactly solves the constraints.

To visualize our results, we employ the following formula for the squared amplitudes of the pole ansatz solutions \eqref{eq:ansatz}, for which a proof is given in Appendix~\ref{app:moduli},   
\begin{equation}\label{eq:moduli}
\left(\begin{array}{c} |u|^2 \\ |v|^2 \end{array}\right)	= B      \left(\begin{array}{c} 1 \\ 1     \end{array}\right)-\ii \sum_{j=1}^N \left(\begin{array}{c} \alpha(x-a_j-\ii\delta/2)-\alpha(x-a_j^*+\ii\delta/2) \\
\alpha(x-a_j+\ii\delta/2)-\alpha(x-a_j^*-\ii\delta/2)     \end{array}\right).
\end{equation}
where
\begin{equation}\label{eq:moduli2}
B\coloneqq  \lambda^2+\bigg(\frac{\pi}{2\delta}\bigg)^2\Bigg|\sum_{j=1}^N c_j \Bigg|^2
\end{equation}
It is interesting to note that the dynamics of $\{c_j\}_{j=1}^N$ play no role in \eqref{eq:moduli}--\eqref{eq:moduli2}: because $\sum_{j=1}^N c_j$ is conserved in time, only the initial values of $\{c_j\}_{j=1}^N$ are required to compute \eqref{eq:moduli2}. We remark that an analogous formula in the context of HNLS solitons was given in \cite[Eq.~(19)]{matsuno2002}.

From \eqref{eq:moduli}--\eqref{eq:moduli2} and the fact that the functions $f_{j,\pm}(x)=-\ii(\alpha(x-a_j\mp\ii\delta/2)-\alpha(x-a_j^*\mp\ii\delta/2))$ decay rapidly (when \eqref{eq:imaj} holds), it is clear that the soliton amplitudes exhibit localized excitations on the (amplitude) background $\sqrt{B}$ about the points $x=\re(a_j)$ for $j=1,\ldots,N$. Moreover, it is straightforward to verify that (when \eqref{eq:imaj} holds) $f_{j,\pm}(x)\lessgtr 0$ is satisfied for $x\in \R$; the $u$-solitons are ``dark" while the $v$-solitons are ``bright." This phenomenon is illustrated in the figures presented below. 

The parameter $\delta$ can be removed in \eqref{eq:CM}, \eqref{eq:ansatz}--\eqref{eq:imaj}, and \eqref{eq:moduli} by the rescalings
\begin{equation}\label{eq:units}	
x\to \delta x,\qquad t\to \delta^2 t, \qquad (u,v)\to (u,v)/\sqrt{\delta},\qquad \lambda\to \lambda/\sqrt{\delta}, \qquad a_j\to \delta a_j\qquad c_j\to \sqrt{\delta}c_j
\end{equation}
(for $j=1,\ldots,N$ in the case of $\{a_j,c_j\}_{j=1}^N$). These transformations provide the natural units for the quantities in question, which we use in our figures.

\subsubsection{One-soliton solutions}\label{subsec:1sol}

When $N=1$, the system of equations \eqref{eq:CM} and \eqref{eq:cjdot} reduces to 
\begin{equation}
\ddot{a}_1=0,\qquad \dot{c}_1=0,
\end{equation}
for which the general solution is
\begin{equation}\label{eq:gensol}
a_1(t)=a_{1,0}+\eta_{1}t,\qquad c_1(t)=c_{1,0}
\end{equation}
for some complex constants $a_{1,0}$, $\eta_1$, and $c_{1,0}$. The requirement \eqref{eq:imaj} imposes the condition $-3\delta/2<\im(a_{1,0})<-\delta/2$. In terms of the parameters in \eqref{eq:gensol}, the conditions \eqref{eq:ajdot} and \eqref{eq:constraint} at $t=0$ read, respectively
\begin{equation}\label{eq:c10constraint1}
|c_{1,0}|^2\bigg(\frac{\lambda}{c_{1,0}^*}-\ii \alpha(a_{1,0}-a_{1,0}^*+\ii\delta)\bigg)=-1,\qquad c_{1,0} \eta_1 =2\lambda,
\end{equation}
which, recalling \eqref{eq:alpha} and the standard identities $\coth(z+\ii\pi/2)=\tanh(z)$, $\tanh(\ii z)=\ii\tan(z)$, imply
\begin{equation}\label{eq:c10constraint2}
|c_{1,0}|^2\bigg(\frac12 \eta_{1}^*+\frac{\pi}{2\delta}\tan\bigg(\frac{\pi}{\delta}\,\im(a_{1,0})\bigg)\bigg)=-1.
\end{equation}
This equation has a positive solution for $|c_{1,0}|$ if and only if $\eta_1\in\R$ and
\begin{equation}\label{eq:ajineq}
\eta_1+\frac{\pi}{\delta}\tan\bigg(\frac{\pi}{\delta}\,\im(a_{1,0})\bigg)<0.
\end{equation}
Suppose $a_{1,0}\in \C$ and $v_{1}\in \R$ satisfying \eqref{eq:imaj} at $t=0$ and \eqref{eq:ajineq} are given. Then $c_{1,0}$ is determined up to a phase by \eqref{eq:c10constraint2}, but the requirement that $\lambda$ is real and the second equation in \eqref{eq:c10constraint1} restrict $c_{1,0}$ to be real. Applying Theorem~\ref{thm:solitons} to the ansatz \eqref{eq:ansatz} with $N=1$, we arrive at the following explicit traveling wave solution for the IMM system,
\begin{equation}\label{eq:explicit1sol}
\left(\begin{array}{c} u(x,t) \\ v(x,t) \end{array}\right)= \frac{c_{1,0}\eta_{1}}{2}\left(\begin{array}{c} 1 \\ -1 \end{array}\right)+\ii c_{1,0} \left(\begin{array}{c} \alpha(x-a_{1,0}-\eta_1 t-\ii\delta/2) \\ -\alpha(x-a_{1,0}-\eta_1 t+\ii\delta/2)   \end{array}\right),
\end{equation}
where
\begin{equation}
c_{1,0}=\frac{\sqrt{2}}{\sqrt{-\big(\eta_{1}+\frac{\pi}{\delta}\tan\big(\frac{\pi}{\delta}\,\im(a_{1,0})\big)\big)}}.
\end{equation}
An example of such a solution is presented in Fig.~\ref{fig:1soliton}. 

\begin{remark}
Note that when $\eta_1=0$, we generate stationary solutions and $\lambda=0$. Note, however, that such solutions still have non-trivial asymptotics (see \eqref{eq:moduli}--\eqref{eq:moduli2}). Via the Galilean transformation \eqref{eq:galilean}, these solutions may be transformed into non-stationary solutions. 
\end{remark}

\begin{figure}
  \includegraphics[width=\linewidth]{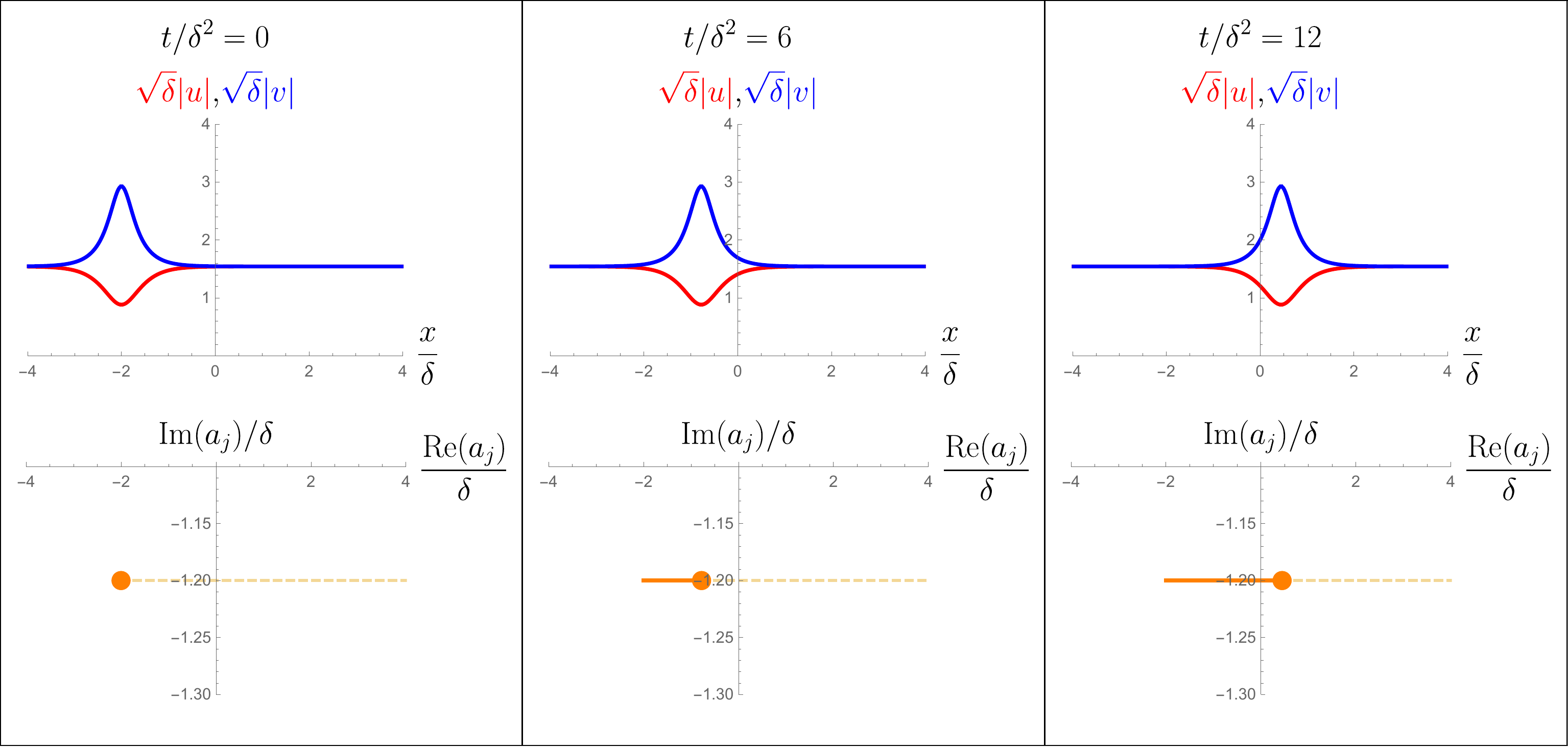}
  \caption{A one-soliton solution with the following (initial) parameter values: $\lambda=1/(10\sqrt{\delta})$, $a_1(0)=-(2+12\ii/10)\delta,$ $c_1(0)\approx 0.98091\sqrt{\delta}$ (the velocity of $a_1$ is determined by \eqref{eq:ajdot}). The resulting dynamics of the parameters is obtained from \eqref{eq:gensol}. The (amplitudes of the) solution to the IMM system \eqref{eq:IMM}--\eqref{eq:TT} is plotted in the upper frames at three points in time. The corresponding dynamics of $a_1$ is plotted in the lower frames, where a dot indicates the value at each time, the bold line shows the trajectory since $t=0$, and the dotted line shows the future trajectory (until it exceeds the range of the plot). }
  \label{fig:1soliton}
\end{figure}

\subsubsection{Two-soliton solutions}\label{subsubsec:twosol}

We first develop a method for generating exact initial data for two-soliton solutions. When $N=2$, the constraints \eqref{eq:constraint} reduce to 
\begin{equation}\label{eq:c1c2}
\begin{split}
&c_{1}\big(\lambda-\ii c_{1}^*\alpha(a_{1}-a_{1}^*+\ii\delta)-\ii c_{2}^*\alpha(a_{1}-a_{2}^*+\ii\delta)\big)+1=0, \\
&c_{2}	\big(\lambda-\ii c_{1}^*\alpha(a_{2}-a_{1}^*+\ii\delta)-\ii c_{2}^*\alpha(a_{2}-a_{2}^*+\ii\delta)\big)+1=0.
\end{split}
\end{equation}

The imaginary parts of \eqref{eq:c1c2} read
\begin{equation}\label{eq:c1c2Im}
\begin{split}
\im(c_{1})\lambda-\re\big(c_{1} c_{2}^*\alpha(a_{1}-a_{2}^*+\ii\delta)\big)=0, \\
\im(c_{2})\lambda-\re\big(c_{1}^* c_{2}\alpha(a_{2}-a_{1}^*+\ii\delta)\big)=0. \\
\end{split}	
\end{equation}
We assume  $\im(c_{1})\neq 0$, in which case the equations in \eqref{eq:c1c2Im} are satisfied provided that
\begin{equation}\label{eq:c1c2d}
c_{2}=-c_{1}+d
\end{equation}
for some $d\in \R$ (i.e., $\im(c_1+c_2)=0$) and that
\begin{equation}\label{eq:lambda}
\lambda=-\frac{\re\big( c_{1}(c_{1}^*-d)\alpha(a_{1}-a_{2}^*+\ii\delta)\big)}{\im(c_{1})}. 	
\end{equation}

Taking the real parts of \eqref{eq:c1c2} with \eqref{eq:c1c2d}--\eqref{eq:lambda} gives the following system,
\begin{align}\label{eq:c1dsystem}
&\frac{\re(c_{1})\re\big( c_{1}(c_{1}^*-d)\alpha(a_{1}-a_{2}^*+\ii\delta)\big)}{\im(c_{1})}+\ii|c_1|^2\alpha(a_1-a_1^*+\ii\delta)+\im\big( c_1 (c^*_1-d)\alpha(a_1-a_2^*+\ii\delta)\big)-1=0, \nonumber \\
&\begin{multlined}
\frac{\big(\re(c_1)-d\big)\re\big( c_1(c_1^*-d)\alpha(a_1-a_2^*+\ii\delta)\big)}{\im(c_1)}-\im\big(c_1^*(c_1-d)\alpha(a_2-a_1^*+\ii\delta)\big)\\
-\ii|c_1-d|^2\alpha(a_2-a_2^*+\ii\delta)+1=0.
\end{multlined}
\end{align}

The first equation in \eqref{eq:c1dsystem} is linear in $d$; provided the coefficient of $d$ in \eqref{eq:c1dsystem} is nonzero, we can solve explicitly for $d$ in terms of $c_1$, $a_1$, and $a_2$.  Substituting the resulting expression for $d$ into the second equation in \eqref{eq:c1dsystem} and clearing out denominators gives a seventh-degree polynomial equation in $\re(c_1)$ with real coefficients given by functions of $a_1$, $a_2$, and $\im(c_1)$. The coefficient of the highest-order term in this polynomial has the form $\im(c_1)f(a_1,a_2)$, where $f(a_1,a_2)$ is nonzero for generic $a_1,a_2$. Thus, we may choose $\im(c_1)\neq 0$ and generic $a_1,a_2$ satisfying \eqref{eq:imaj}--\eqref{eq:ajak} to obtain a seventh-degree polynomial in $\re(c_1)$ with real coefficients, which is guaranteed to have at least one real root. When this root together with our chosen values of $\im(c_1)$ and $a_1,a_2$ corresponds to a nonzero coefficient of $d$ in the first equation in \eqref{eq:c1dsystem}, we have obtained a solution of the constraints \eqref{eq:c1c2}, and admissible initial data for Theorem~\ref{thm:solitons}. An example of such a solution is provided in Figs.~\ref{fig:2soliton}--\ref{fig:2solitoncs}.

\begin{figure}
  \includegraphics[width=\linewidth]{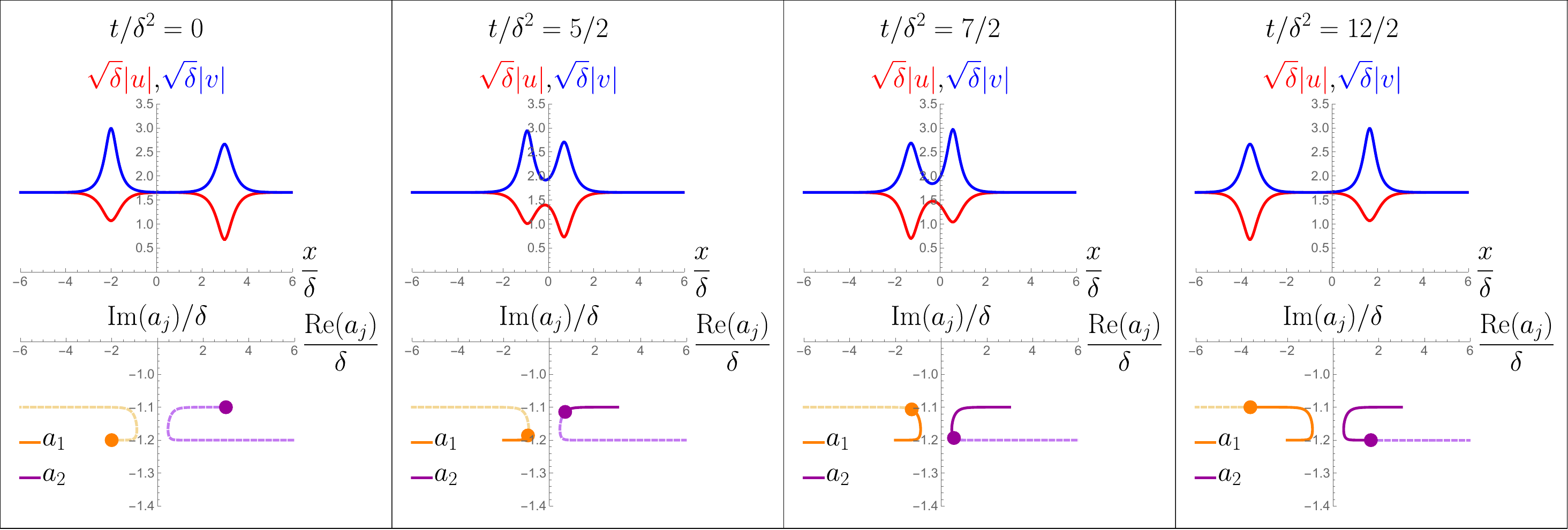}
  \caption{A two-soliton solution with the following (initial) parameter values: $\lambda\approx1.6376/\sqrt{\delta}$, $a_1(0)=-(2+12\ii/10)\delta$, $a_2(0)=(3-11\ii/10)\delta$, $c_1(0)\approx (0.29908 + \ii)\sqrt{\delta}$, $c_2(0)\approx -(0.14219 +  \ii)\sqrt{\delta}$ (the initial velocities of $a_1,a_2$ are determined by \eqref{eq:ajdot}). The resulting dynamics of the parameters is obtained from \eqref{eq:CM} and \eqref{eq:cjdot}. The (amplitudes of the) solution to the IMM system is plotted in the upper frames at four points in time. The corresponding dynamics of $a_1,a_2$ is plotted in the lower frames, where dots indicate the values at each time, the bold lines show the trajectories since $t=0$, and the dotted line shows the future trajectories (until they exceed the range of the plot).}
  \label{fig:2soliton}
\end{figure}

\begin{figure}
\centering
  \includegraphics[width=7cm]{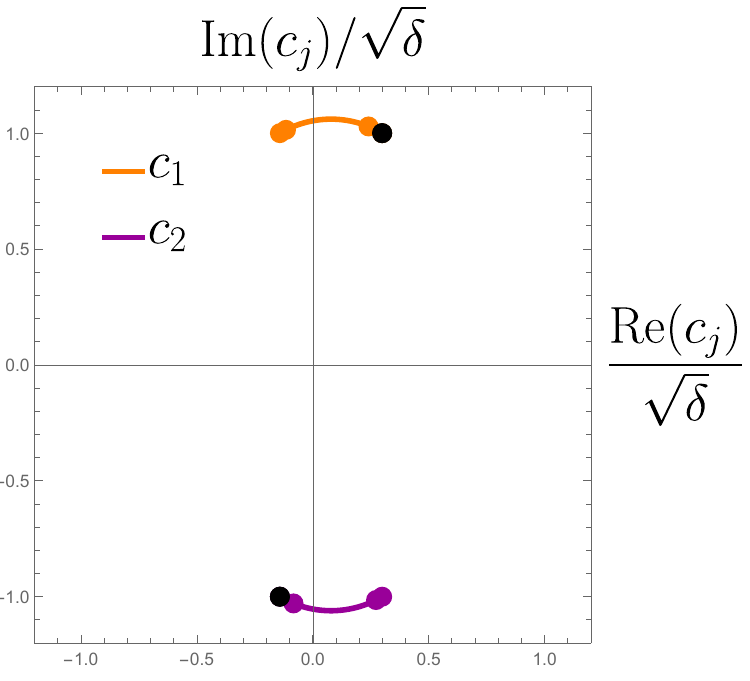}
  \caption{The trajectories of the parameters $c_1,c_2$ in the two-soliton solution in Fig.~\ref{fig:2soliton}. The black dots indicate the values at $t=0$, while the colored dots indicate the values at the subsequent times depicted in Fig.~\ref{fig:2soliton}. Note that some colored dots are very close or coincident, as $c_1,c_2$ only move appreciably when $a_1,a_2$ are close together.}
  \label{fig:2solitoncs}
\end{figure}

\subsubsection{Three-soliton solutions}\label{subsubsec:threesol}

The ideas from Section~\ref{subsubsec:twosol} can be extended to treat higher-$N$ cases, though the systems of polynomial equations become more complicated and we only present key details. When $N=3$, the constraints \eqref{eq:constraint} reduce to 
\begin{equation}\label{eq:c1c2c3}
\begin{split}
&c_1\big(\lambda-\ii c_1^*\alpha(a_1-a_1^*+\ii\delta)-\ii c_2^*\alpha(a_1-a_2^*+\ii\delta)-\ii c_3^*\alpha(a_1-a_3^*+\ii\delta)\big)+1=0, \\
&c_2	\big(\lambda-\ii c_1^*\alpha(a_2-a_1^*+\ii\delta)-\ii c_2^*\alpha(a_2-a_2^*+\ii\delta)-\ii c_3^*\alpha(a_2-a_3^*+\ii\delta)    \big)+1=0, \\
&c_3	\big(\lambda-\ii c_1^*\alpha(a_3-a_1^*+\ii\delta)-\ii c_2^*\alpha(a_3-a_2^*+\ii\delta)-\ii c_3^*\alpha(a_3-a_3^*+\ii\delta)   \big)+1=0.
\end{split}
\end{equation}

Similarly to before, by considering linear combinations of the imaginary parts of \eqref{eq:c1c2c3}, it may be shown that
\begin{equation}\label{eq:c3c1c2}
c_3=-(c_1+c_2)+d	
\end{equation}
for some $d\in \R$ (i.e., $\im(c_1+c_2+c_3)=0$ and, provided $\im(c_1+c_2)\neq 0$, 
\begin{equation}\label{eq:lambdac1c2c3}
\lambda=-\frac{\re\big(c_1(c_1^*+c_2^*-d)\alpha(a_1-a_3^*+\ii\delta)+c_2(c_1^*+c_2^*-d)\alpha(a_2-a_3^*+\ii\delta)  \big)}{\im(c_1+c_2)}	
\end{equation}
must hold.

By substituting \eqref{eq:c3c1c2}--\eqref{eq:lambdac1c2c3} into \eqref{eq:c1c2c3}, clearing out denominators, and taking the real and imaginary parts of each equation, we obtain a system of four independent polynomial equations in the variables $\re(c_1)$, $\re(c_2)$, $\im(c_1)$, $\im(c_2)$, and $d$ with real coefficients. Once $a_1,a_2,a_3$ are specified, these equations, together with the requirement that $\im(c_1+c_2)\neq 0$, can be solved by standard methods (when a solution exists). In Figs.~\ref{fig:3soliton}--\ref{fig:3solitoncs}, we provide an example of a three-soliton solution generated using this procedure.  

\begin{figure}
  \includegraphics[width=\linewidth]{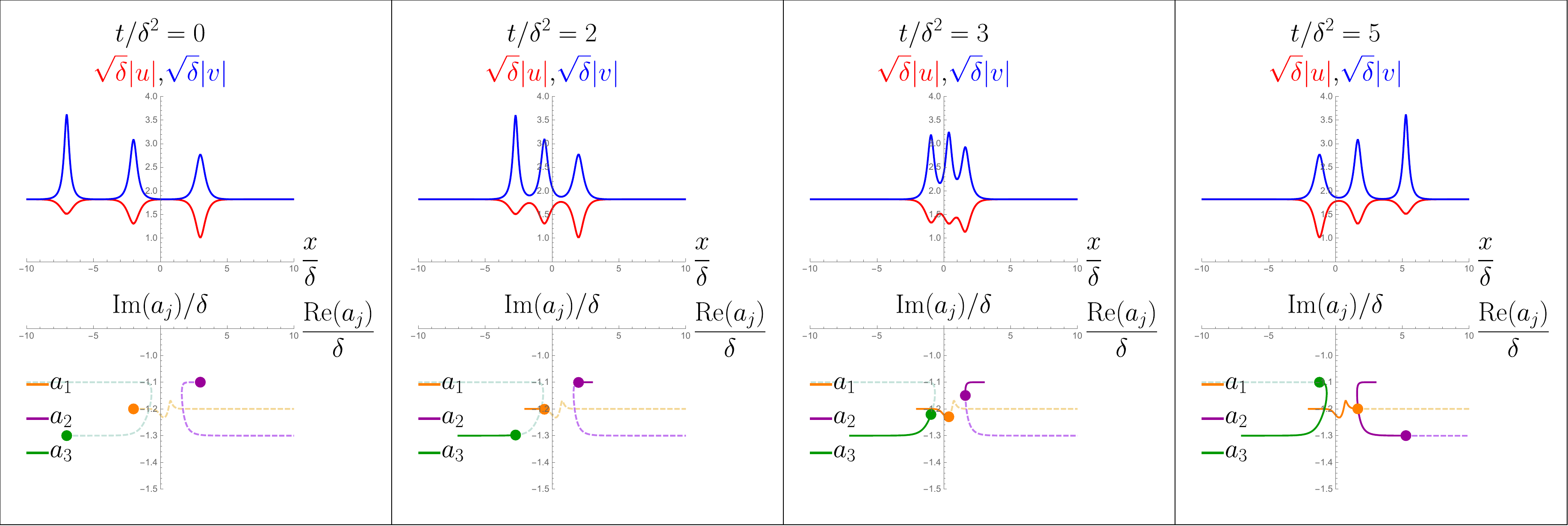}
  \caption{A three-soliton solution with the following (initial) parameter values: $\lambda\approx -1.1187/\sqrt{\delta}$, $a_1(0)=-(2+12\ii/10)\delta$, $a_2(0)=(3-11\ii/10)\delta$, $a_3(0)=-(7+13\ii/10)\delta,c_1(0)\approx (0.81989 - 0.77721\ii)\sqrt{\delta}$, $c_2(0)\approx (-0.77840 + 0.83807\ii)\sqrt{\delta}$, $c_3(0)\approx -(0.95540 + 0.060857\ii)\sqrt{\delta}$ (the initial velocities of $a_1,a_2,a_3$ are determined by \eqref{eq:ajdot}). The resulting dynamics of the parameters is obtained from \eqref{eq:CM} and \eqref{eq:cjdot}. The (amplitudes of the) solution to the IMM system \eqref{eq:IMM}--\eqref{eq:TT} is plotted in the upper frames at four points in time. The corresponding dynamics of $a_1,a_2,a_3$ is plotted in the lower frames, where dots indicate the values at each time, the bold lines show the trajectories since $t=0$, and the dotted lines show the future trajectories (until they exceed the range of the plot).}
  \label{fig:3soliton}
\end{figure}

\begin{figure}
\centering
  \includegraphics[width=7cm]{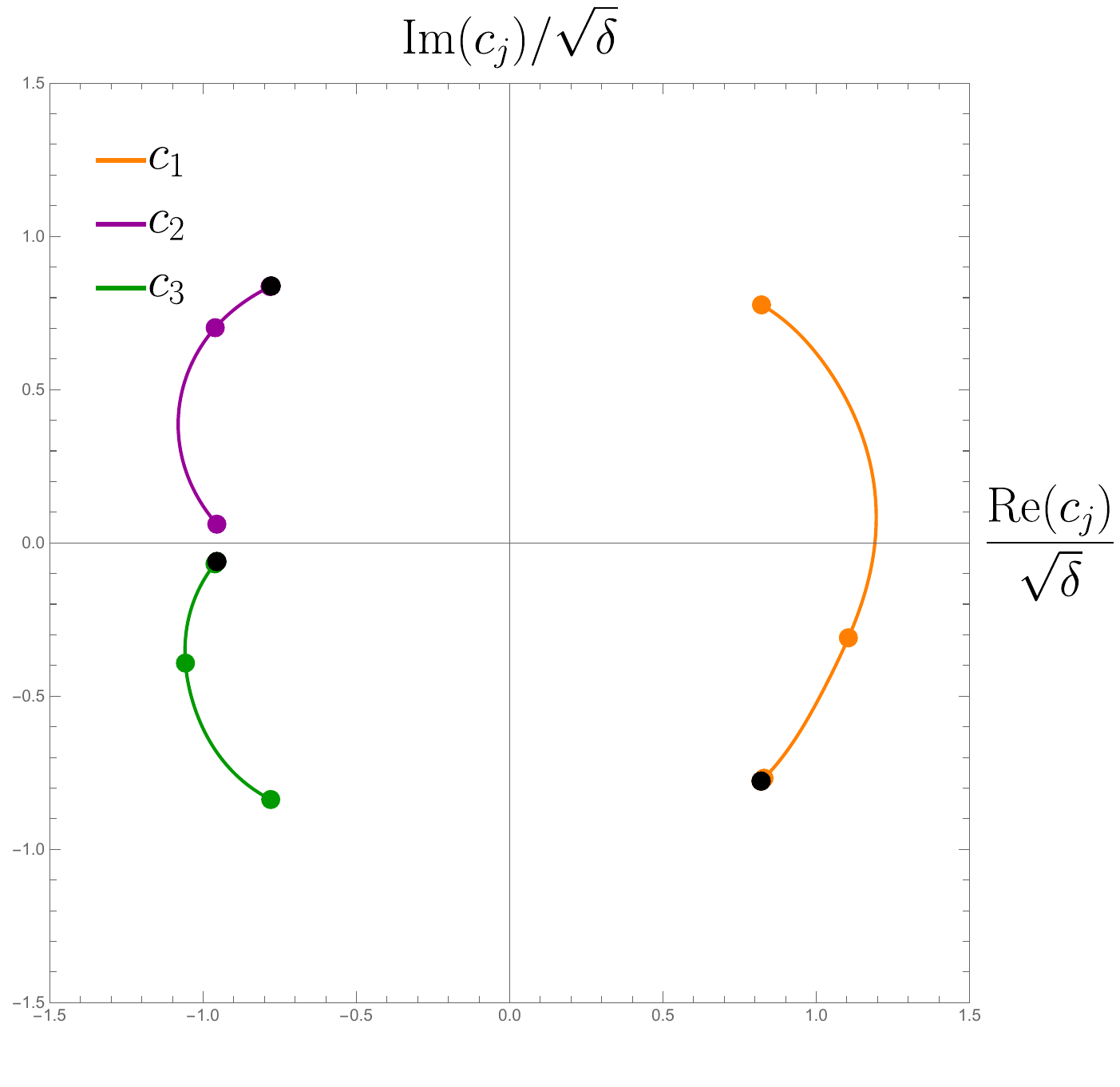}
  \caption{The trajectories of the parameters $c_1,c_2,c_3$ in the three-soliton solution in Fig.~\ref{fig:3soliton}. The black dots indicate the values at $t=0$, while the colored dots indicate the values at the subsequent times depicted in Fig.~\ref{fig:3soliton}. Note that some colored dots are very close or coincident, as each $c_j$ $(j=1,2,3)$ only moves appreciably when its corresponding parameter $a_j$ is close to at least one other $a_k$ $(k\neq j)$. }
  \label{fig:3solitoncs}
\end{figure}

\section{Periodic case}\label{sec:periodic}

We consider the IMM system \eqref{eq:IMM} with spatially periodic boundary conditions $u(x+2\ell,t)=u(x,t)$ and $v(x+2\ell,t)=v(x,t)$ for some $\ell>0$. In this case, the $T$ and $\tilde{T}$ operators \eqref{eq:TT} can be represented as integral operators on the interval $[-\ell,\ell)$ with a certain $2\ell$-periodic kernel; we refer to \cite{berntsonlangmann2021} for justification of this claim in a similar context. We define the periodic IMM system to be \eqref{eq:IMM} with 
\begin{equation}\label{eq:TTe}
\begin{split}
({Tf})(x)\coloneqq &\; \frac{1}{\pi}\pvint_{-\ell}^{\ell} \zeta_1(x'-x;\ell,\ii\delta)f(x')\,\mathrm{d}x', \\
({\tilde{T}f})(x)\coloneqq &\;\frac{1}{\pi}\int_{-\ell}^{\ell} \zeta_1(x'-x+\ii\delta;\ell,\ii\delta)f(x')\,\mathrm{d}x' ,
\end{split}
\end{equation}
where 
\begin{equation}\label{eq:zeta1}
\zeta_1(z;\ell,\ii\delta)\coloneqq  \zeta(z;\ell,\ii\delta)-\frac{\zeta(\ell;\ell,\ii\delta)}{\ell}z,
\end{equation}
with $\zeta(z;\ell,\ii\delta)$ the Weierstrass $\zeta$-function with half-periods $\ell$ and $\ii\delta$. 

We will adapt the main result of Section~\ref{sec:solitons} to the spatially periodic setting. A key object in doing so is the elliptic CM system, Case IV of \eqref{eq:CM}--\eqref{eq:V}. In order to construct solutions of the periodic IMM system controlled by the elliptic CM system, we modify the ansatz \eqref{eq:ansatz} in a number of ways. 

We replace the hyperbolic special function $\alpha(z)$ defined in \eqref{eq:alpha} with a particular elliptic generalization,
\begin{equation}\label{eq:alphaelliptic}
\alpha(z)\coloneqq \zeta_2(z;\ell,\ii\delta),
\end{equation}
where
\begin{equation}\label{eq:zeta2}
\zeta_2(z;\ell,\ii\delta)\coloneqq \zeta(z;\ell,\ii\delta)-\frac{\zeta(\ii\delta;\ell,\ii\delta)}{\ii\delta}z.
\end{equation}
We further define
\begin{equation}\label{eq:wp2}
\wp_2(z;\ell,\ii\delta)\coloneqq -\partial_z \zeta_2(z;\ell,\ii\delta)=\wp(z;\ell,\ii\delta)+\frac{\zeta(\ii\delta;\ell,\ii\delta)}{\ii\delta};
\end{equation}
note that $V(z)$ in Case IV of \eqref{eq:V} is thus equal to $-\alpha'(z)$. Basic properties of and identities for these special functions can be found in Appendix~\ref{app:functional}. 

Our main result on the periodic IMM system is that it admits solutions in the form
\begin{equation}\label{eq:ansatzelliptic}
\left(\begin{array}{c} u(x,t) \\ v(x,t) \end{array}\right) = \ee^{2\ii N \cc t}  \Bigg( \lambda(t) \left(\begin{array}{c} 1 \\ -1 \end{array}\right)  + \ii\sum_{j=1}^N c_j(t) \left(\begin{array}{c} \alpha(x-a_j(t)-\ii\delta/2) \\ -\alpha(x-a_j(t)+\ii\delta/2)   \end{array}\right)\Bigg),
\end{equation}
where 
\begin{equation}\label{eq:gamma0}
\cc\coloneqq \frac{\pi}{2\ell \delta};	
\end{equation}
observe that $\lambda$ has been promoted to a dynamical quantity. 

More specifically, we will show that the periodic IMM system has solutions in the form \eqref{eq:ansatzelliptic} with the parameters satisfying the following conditions: $\{a_j\}_{j=1}^N$ solve the elliptic CM system, Case IV of \eqref{eq:CM}--\eqref{eq:V}, $\{c_j\}_{j=1}^N$ satisfy the constraint\footnote{If the constraint \eqref{eq:cjconstraint} is satisfied at $t=0$, it satisfied at future times; see Lemma~\ref{lem:cjconstraint} for details. Hence, in Theorem~\ref{thm:ellipticsolitons}, \eqref{eq:cjconstraint} is only imposed at $t=0$.}
\begin{equation}\label{eq:cjconstraint}
\sum_{j=1}^N c_j =0
\end{equation}
and solve the ODEs \eqref{eq:cjdot} with $V(z)$ in Case IV of \eqref{eq:V}, and $\lambda$ solves
\begin{equation}\label{eq:lambdadot}
\dot{\lambda}=\frac12\sum_{j=1}^N\sum_{k\neq j}^N (c_j-c_k)\ftwo'(a_j-a_k),
\end{equation}
where
\begin{equation}\label{eq:f2}
\ftwo(z)\coloneqq \alpha(z)^2-V(z), 
\end{equation}
with $\alpha(z)$ and $V(z)$
 given by \eqref{eq:alphaelliptic} and Case IV of \eqref{eq:V}, respectively. (As in the case of Theorem~\ref{thm:solitons}, there are constraints on the initial conditions of these parameters as well as technical conditions on the behavior of $\{a_j\}_{j=1}^N$; see Theorem~\ref{thm:ellipticsolitons} below for the precise statement of the result described above.)
 
 It is interesting to note that the constraint \eqref{eq:cjconstraint} means that no nontrivial traveling wave solutions with $N=1$ can be obtained from the ansatz \eqref{eq:ansatzelliptic}; correspondingly, the case $N=2$ of \eqref{eq:ansatzelliptic} with \eqref{eq:cjconstraint} is the simplest nontrivial one. 

The precise statement of the result described above is given in Section~\ref{subsec:resultelliptic} as Theorem~\ref{thm:ellipticsolitons} and the corresponding proof can be found in Appendix~\ref{app:proofs}. Analysis of the $N=2,3$ cases of Theorem~\ref{thm:ellipticsolitons} is performed in Section~\ref{subsec:examplesperiodic}, where visualizations of particular solutions are also provided. 

\subsection{Result}\label{subsec:resultelliptic}

Our result is stated and followed by several remarks and a discussion of its relation to Theorem~\ref{thm:solitons}. We emphasize that $\alpha(z)$ and $V(z)$ are given by  \eqref{eq:alphaelliptic} and Case IV of \eqref{eq:V}, respectively. 

\begin{theorem}\label{thm:ellipticsolitons}
For $N\in \Z_{\geq 2}$, let $\lambda$ and $\{a_j,c_j\}_{j=1}^N$ be a solution of the system of ODEs consisting of \eqref{eq:CM}, \eqref{eq:cjdot}, and \eqref{eq:lambdadot} on an interval $[0,\tau)$ for some $\tau\in (0,\infty)\cup\{\infty\}$ and with initial conditions that satisfy
\begin{equation}\label{eq:constraintperiodic}
c_j \Bigg(\lambda^*-\ii\sum_{k=1}^N c_k^*\alpha(a_j-a_k^*+\ii\delta)\Bigg)+1=0	\quad (j=1,\ldots,N)
\end{equation}
 and \eqref{eq:cjconstraint} at $t=0$. Moreover, suppose that the conditions \eqref{eq:imaj} and 
 \begin{equation}\label{eq:ajakelliptic}
a_j\neq a_k \bmod 2\ell \quad (1\leq j<k\leq N)
 \end{equation}
 are satisfied for $t\in [0,\tau)$. Then, \eqref{eq:ansatzelliptic} solves the periodic IMM system, \eqref{eq:IMM} with \eqref{eq:TTe}, on $[0,\tau)$. 
\end{theorem}

\subsubsection{Remarks on Theorem~\ref{thm:ellipticsolitons}}\label{subsubsec:remarksperiodic}

We give several comments on our result, some of which highlight differences versus Theorem~\ref{thm:solitons}. 

\begin{enumerate}
\item The ansatz \eqref{eq:ansatz} with \eqref{eq:alphaelliptic} is not  $2\ell$-periodic for generic $\{c_j\}_{j=1}^N$ because $\zeta_2(z)$ is only quasi-$2\ell$-periodic \eqref{eq:realperiod} and this necessitates the constraint \eqref{eq:cjconstraint}.  
\item The derivatives of the functions $A_{\pm}(z)$ in \eqref{eq:Apm} with $\alpha(z)$ as in \eqref{eq:alphaelliptic} are no longer eigenfunctions of the periodic $\cT$ operator \eqref{eq:cT} with \eqref{eq:TTe}, but instead satisfy the relations \cite{berntson2022elliptic}
\begin{equation}\label{eq:TA}
(\cT A_{+}'(\cdot-a_j))(x)=\ii A_+'(x-a_j)+2\ii \cc \left(\begin{array}{c} 0 \\ 1 \end{array}\right), \qquad \cT (A_{-}'(\cdot-a_j))(x)=-\ii A_-'(x-a_j^*)+2\ii \cc \left(\begin{array}{c} 1 \\ 0 \end{array}\right).
\end{equation}
The prefactor $\ee^{2\ii N\cc t}$ is required to make the ansatz \eqref{eq:ansatzelliptic} consistent in the presence of new terms proportional to $\gamma_0$ generated through \eqref{eq:TA} (see Proposition~\ref{prop:ansatz} and its proof in Appendix~\ref{app:proofs} for details). 

\item The proof of Theorem~\ref{thm:ellipticsolitons} involves the use of functional identities for $\alpha(z)$ and $V(z)$ at key points. These identities involve more terms in the elliptic case than in the hyperbolic case (see Appendix~\ref{app:functional} for details). To make the ansatz \eqref{eq:ansatzelliptic} consistent in the presence of these new terms, we promote $\lambda$ to a complex dynamical quantity (see Proposition~\ref{prop:ansatz} and its proof in Appendix~\ref{app:proofs} for details).
\end{enumerate}

\subsubsection{The relation between Theorems~\ref{thm:solitons} and \ref{thm:ellipticsolitons}} \label{subsubsec:limit}

We described the key changes to adapt Theorem~\ref{thm:solitons} to the spatially periodic setting in Section~\ref{subsubsec:remarksperiodic}. To recover a special case of Theorem~\ref{thm:solitons} from Theorem~\ref{thm:ellipticsolitons}, one takes a limit as $\ell\to \infty$ in Theorem~\ref{thm:ellipticsolitons}. In this limit, the periodic IMM system \eqref{eq:IMM} with \eqref{eq:TTe} becomes the real-line IMM system \eqref{eq:IMM} with \eqref{eq:TT} and the elliptic functions $\alpha(z)$ and $V(z)$ degenerate to their hyperbolic counterparts (see Appendix~\ref{subsec:hyperbolic} for details). It is then straightforward to verify that Theorem~\ref{thm:ellipticsolitons} essentially becomes Theorem~\ref{thm:solitons} (note that \eqref{eq:ftwolimit} implies that \eqref{eq:lambdadot} becomes $\dot{\lambda}=0$; we may assume $\lambda$ is a real constant without loss of generality by \eqref{eq:U1timesU1}). The outstanding issue is the constraint \eqref{eq:cjconstraint}, which is essential in the periodic case but superfluous (though not inconsistent) in the real-line case. Thus, in the limit described above, Theorem~\ref{thm:ellipticsolitons} becomes a specialization of Theorem~\ref{thm:solitons} with the additional constraint \eqref{eq:cjconstraint}. 

Practically, to prove Theorem~\ref{thm:solitons}, one takes the proof of Theorem~\ref{thm:ellipticsolitons} in Appendix~\ref{app:proofs}, drops the condition \eqref{eq:cjconstraint}, and makes the replacements $\ftwo(z)\to (\pi/2\delta)^2$ \eqref{eq:ftwolimit} and $\gamma_0\to 0$ throughout.

\subsection{Examples of solutions}\label{subsec:examplesperiodic}

Analogously to Section~\ref{subsec:examples}, where methods for generating and examples of soliton solutions are provided, we consider the cases $N=2,3$ of Theorem~\ref{thm:ellipticsolitons} (recall that the case $N=1$ does not exist). The key step, as before, is solving the nonlinear constraints \eqref{eq:constraintperiodic}, now subject to \eqref{eq:cjconstraint}. The $N=2$ case is solved in full, giving an algorithm for finding all solutions of the constraints \eqref{eq:cjconstraint} with \eqref{eq:constraintperiodic} when $N=2$. In the case $N=3$, we manipulate the constraints into a manageable form and provide a method for generating a restricted class of solutions. We also provide visualizations of solutions of the periodic IMM system resulting from these methods. 

The squared amplitudes of the solutions of the periodic IMM system constructed by Theorem~\ref{thm:ellipticsolitons} are given by
\begin{equation}\label{eq:moduliperiodic}
\left(\begin{array}{c} |u|^2 \\ |v|^2 \end{array}\right)	= B       \left(\begin{array}{c} 1 \\ 1     \end{array}\right)-\ii \sum_{j=1}^N \left(\begin{array}{c} \alpha(x-a_j-\ii\delta/2)-\alpha(x-a_j^*+\ii\delta/2) \\
\alpha(x-a_j+\ii\delta/2)-\alpha(x-a_j^*-\ii\delta/2)     \end{array}\right),
\end{equation}
where
\begin{equation}\label{eq:B(t)}
B(t)\coloneqq |\lambda(t)|^2+\frac12\sum_{j=1}^N\sum_{k=1}^N c_j(t)c_k(t)^* \ftwo \big(a_j(t)-a_k(t)^*+\ii\delta \big);	
\end{equation}
see Appendix~\ref{app:moduli} for a proof. As in the real-line case, $B$ is apparently dynamical, but actually conserved in time; see Appendix~\ref{app:B} for a proof. Thus, the dynamics of $\lambda$ and $\{c_j\}_{j=1}^N$ play no role in \eqref{eq:moduliperiodic}--\eqref{eq:B(t)}.

As in the real-line case, our solutions exhibit\footnote{To see this, one again verifies that (when \eqref{eq:imaj} holds) $f_{j,\pm}(x)=-\ii(\alpha(x-a_j\mp\ii\delta/2)-\alpha(x-a_j^*\pm\ii\delta/2))\lessgtr 0$ is satisfied for $x\in [-\ell,\ell)$.} ``dark" excitations in the $u$-amplitude and ``bright" excitations in the $v$-amplitude. We provide visualizations of example solutions below with units obtained from \eqref{eq:units} and $\ell$ expressed in units of $\delta$.\footnote{To see that \eqref{eq:units} is still applicable when $\alpha(z)$ and $V(z)$ are given by \eqref{eq:alphaelliptic} and Case IV of \eqref{eq:V}, respectively, we recall the scaling formulas for the Weierstrass $\zeta$- and $\wp$-functions \cite[Eqs.~(23.10.17--23.10.18)]{DLMF},
\begin{equation}\label{eq:Weierstrassscalings}
\zeta(cz;c\ell,c\delta)=c^{-1}\zeta(z;\ell,\delta) ,\qquad \wp(cz;c\ell,c\delta)=c^{-2}\wp(z;\ell,\delta),	
\end{equation}
valid for arbitrary $c\in \C\setminus\{0\}$ and observe that \eqref{eq:Weierstrassscalings} extends to the $\zeta_2$- and $\wp_2$-functions defined in \eqref{eq:zeta2} and \eqref{eq:wp2}, respectively. 
   }   
\subsubsection{Two-wave solutions}\label{subsec:2solelliptic}

We consider the constraints \eqref{eq:constraint} with $N=2$ at some fixed time; by imposing \eqref{eq:cjconstraint} as $c_2=-c_1$, using the $\mathrm{U}(1)\times \mathrm{U}(1)$ invariance of the IMM system \eqref{eq:U1timesU1}, we may assume $\lambda\in \R$ (at this fixed time). The constraints are thus given by \eqref{eq:c1c2} (with $\alpha(z)$ defined in \eqref{eq:alphaelliptic}), which we rearrange to 
\begin{equation}\label{eq:N=2constraints}
\begin{split}
&|c_{1}|^2\bigg(\frac{\lambda}{c_{1}^*}-\ii\big(\alpha(a_{1}-a_{1}^*+\ii\delta)-\alpha(a_{1}-a_{2}^*+\ii\delta)\big)\bigg)=-1, \\
&|c_{1}|^2\bigg(\frac{\lambda}{c_{1}^*}+\ii\big(\alpha(a_{2}-a_{1}^*+\ii\delta)-\alpha(a_{2}-a_{1}^*+\ii\delta)\big)\bigg)=1.
\end{split}
\end{equation}
By adding and subtracting the equations in \eqref{eq:N=2constraints}, we obtain
\begin{equation}\label{eq:N=2constraints2}
\begin{split}
&\ii\big(\alpha(a_{1}-a_{1}^*+\ii\delta)+\alpha(a_{2}-a_{2}^*+\ii\delta)-\alpha(a_{1}-a_{2}^*+\ii\delta)-\alpha(a_{2}-a_{1}^*+\ii\delta)\big)=\frac{2}{|c_{1}|^2}, \\
&\frac{\lambda}{c_{1}^*}=\frac{\ii}{2}\big(\alpha(a_{1}-a_{1}^*+\ii\delta)-\alpha(a_{2}-a_{2}^*+\ii\delta)-\alpha(a_{1}-a_{2}^*+\ii\delta)+\alpha(a_{2}-a_{1}^*+\ii\delta)\big).
\end{split}
\end{equation}
The first equation in \eqref{eq:N=2constraints2} is equivalent to 
\begin{equation}\label{eq:N=2constraints3}
g(\re(a_{1}-a_{2}))=\ii\big(\alpha(2\ii\,\im(a_{1})+\ii\delta)+\alpha(2\ii\,\im(a_{2})+\ii\delta)\big)-\frac{2}{|c_{1}|^2},
\end{equation}
where
\begin{equation}\label{eq:g}
g(z)\coloneqq \ii\big(\alpha(z+\ii\,\im(a_{1}+a_{2})+\ii\delta)-\alpha(z-\ii\,\im(a_{1}+a_{2})+\ii\delta)\big).
\end{equation}
We note that both sides of \eqref{eq:N=2constraints3} are real by the invariance of $\alpha(z)$ under Schwarz conjugation \eqref{eq:Schwarz}. We suppose the imaginary parts of $a_{1},a_{2}$ satisfying \eqref{eq:imaj} are given and investigate the solvability of \eqref{eq:N=2constraints3} for $\re(a_{1}-a_{2})$. To do this, we first analyze the function $g(z)$. 

The function $g(z)$ is real-valued for real arguments and $2\ell$-periodic by \eqref{eq:realperiod}. We determine its extrema for $z\in [0,2\ell)$. The derivative of $g(z)$ is found, using \eqref{eq:alphaelliptic} and \eqref{eq:wp2}, to be
\begin{equation}\label{eq:gprime}
g'(z)=-\ii\big(V(z+\ii\,\im(a_{1}+a_{2})+\ii\delta)-V(z-\ii\,\im(a_{1}+a_{2})+\ii\delta)\big),
\end{equation}
which is clearly an elliptic function. In the nondegenerate case, where $2\ii\,\im(a_{1}+a_{2})\neq 0 \bmod 2\ii\delta$ (otherwise $g(z)$ is the zero function), $g'(z)$ is a degree-four elliptic function, taking each value four times, counting multiplicity, within a period parallelogram. By inspecting \eqref{eq:gprime}, we see that within a period parallelogram, the set of zeroes of $g'(z)$ is $\{0,\ell,\ii\delta,\ell+\ii\delta\}$, all of which are simple. By inserting the first two of these into \eqref{eq:g}, we find that
\begin{equation}\label{eq:N=2constraints4}
0\leq  \sigma g(x)  \leq \big|2\ii\alpha(\ii\,\im(a_{1}+a_{2})+\ii\delta)\big| \quad (x\in \R),
\end{equation}
where 
\begin{equation}\label{eq:sigma}
\sigma = \mathrm{sgn}(2\ii\alpha(\ii\,\im(a_{1}+a_{2})+\ii\delta)).
\end{equation}
The equation \eqref{eq:N=2constraints2} is solvable when its right-hand side lies in the range of $g(x)$,
\begin{equation}\label{eq:N=2constraints5}
0\leq  \sigma\bigg(\ii\big(\alpha(2\ii\,\im(a_{1})+\ii\delta)+\alpha(2\ii\,\im(a_{2})+\ii\delta)\big)- \frac{2}{|c_{1}|^2}\bigg)    \leq \big|2\ii\alpha(\ii\,\im(a_{1}+a_{2})+\ii\delta)\big|.
\end{equation}

We now make use of the observations above. Suppose $\im(a_{1})$, $\im(a_{2})$ and $|c_{1}|$ are chosen so that \eqref{eq:N=2constraints5} holds. Then, \eqref{eq:N=2constraints3} can be solved for $\re(a_{1}-a_{2})$. Using this solution, the first equation in \eqref{eq:N=2constraints2} holds. The second equation in \eqref{eq:N=2constraints2} is then solved after choosing a phase for $c_{1}$ such that the solution for $\lambda$ is real. This provides admissible initial data for Theorem~\ref{thm:ellipticsolitons}; an example solution corresponding to such initial data is given in Figs.~\ref{fig:2sol_ell_1}--\ref{fig:2sol_ell_clambda_1}.  

\begin{remark}
When $N=2$ and $c_2=-c_1$ is imposed, the elliptic CM system, Case IV of \eqref{eq:CM}--\eqref{eq:V} is reduced to
\begin{equation}\label{eq:eom2wave}
\ddot{a}_1=-\ddot{a}_2=-4V'(a_1-a_2).
\end{equation}
By introducing the variables $a_{\pm}\coloneqq a_1\pm a_2$,  the first equation in \eqref{eq:eom2wave} can be written as the system $\ddot{a}_+=0$ and $\ddot{a}_-=-8V'(a_-)$, the latter of which is solved by the sixth Painlev\'{e} transcendent with particular parameter values; see \cite{manin1998} for details.
\end{remark}

\begin{figure}
  \includegraphics[width=\linewidth]{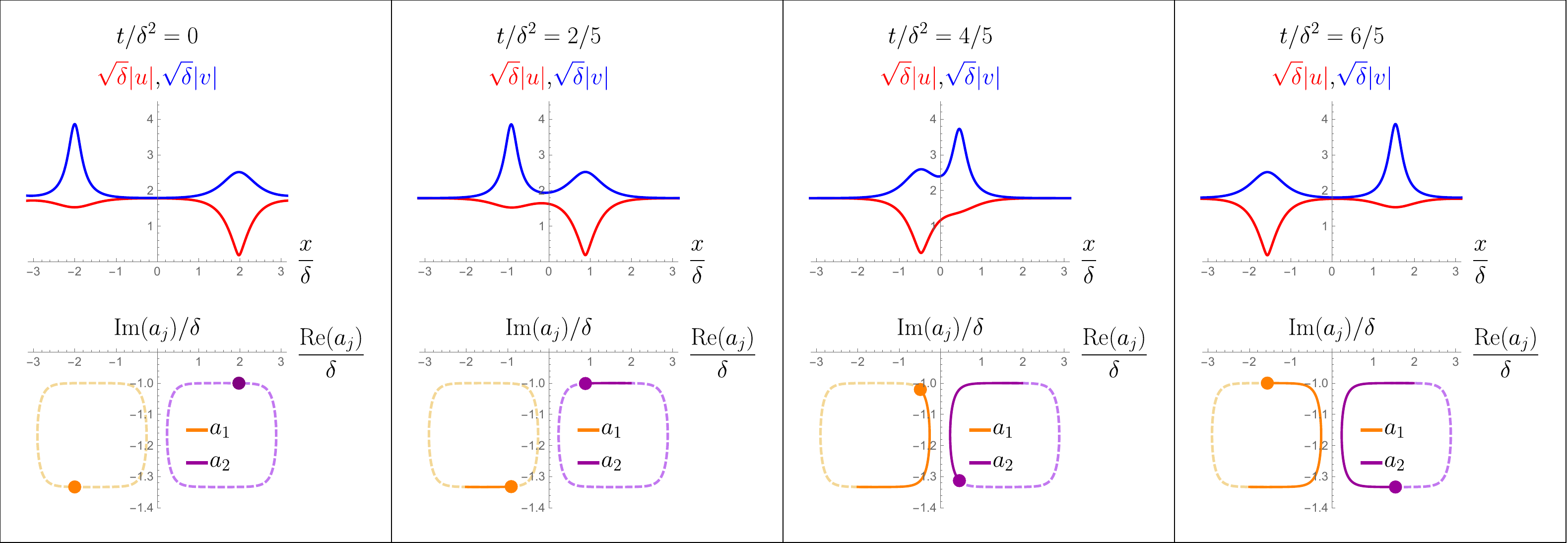}
  \caption{A periodic two-wave solution with $\ell=\pi\delta$ and the initial parameter values $\lambda(0)\approx 1.7836/\sqrt{\delta}$, $a_{1}(0)=-(2+4\ii/3)\delta$, $a_{2}(0)\approx (1.9780 - \ii)\delta$,  $c_{1}(0)=-c_{2}(0)\approx(0.56146 + 0.64878\ii)\sqrt{\delta}$ (the initial velocities of $a_1,a_2$ are determined by \eqref{eq:ajdot}. The resulting dynamics of the parameters are obtained from \eqref{eq:CM} and \eqref{eq:cjdot}. The (amplitudes of the) solution \eqref{eq:ansatz} to the periodic IMM system \eqref{eq:IMM} with \eqref{eq:TTe} is plotted in the upper frames at four points in time. The corresponding dynamics of $a_1,a_2$ is plotted in the lower frames, where dots indicate the values at each time, the bold lines show the trajectories since $t=0$, and the dotted lines show the future trajectories (which, in this example, are time-periodic).
}
  \label{fig:2sol_ell_1}
\end{figure}

\begin{figure}
  \includegraphics[width=\linewidth]{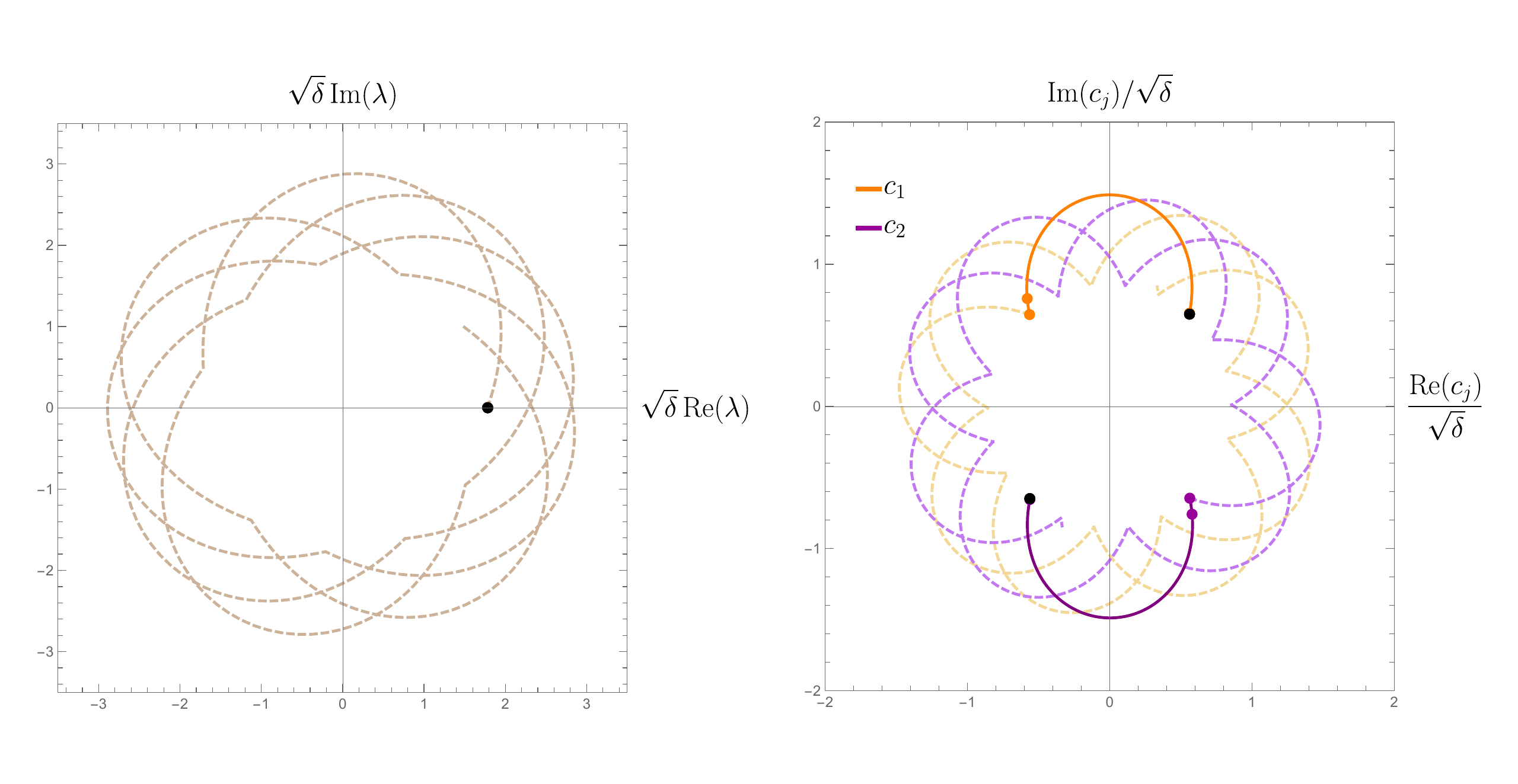}
  \caption{The trajectories of the parameters $\lambda$ and $c_1,c_2$  in the periodic two-wave solution in Fig.~\ref{fig:2sol_ell_1}. The black dots indicate the positions of $\lambda$ and $c_1,c_2$ at $t=0$. The parameter $\lambda$ is effectively static (equal to its initial value) on the time interval depicted in Fig.~\ref{fig:2sol_ell_1}, but the dotted line shows its future trajectory up to the maximal time $t=20 \delta^2$.  For the parameters $c_1,c_2$, the colored dots indicate the values at the subsequent times depicted in Fig.~\ref{fig:2sol_ell_1}, the bold lines show the trajectories up to the last time in Fig.~\ref{fig:2sol_ell_1} ($t=6\delta^2/5$), and the dotted lines show the future trajectories up to the maximal time $t=10\delta^2$.
}
  \label{fig:2sol_ell_clambda_1}
\end{figure}

\subsubsection{Three-wave solutions}

Similar to the three-soliton case considered in Section~\ref{subsubsec:threesol}, the question of solving constraints to generate three-wave initial data is significantly more difficult than its two-wave counterpart. We present a method to generate a restricted class of $N=3$ solutions from Theorem~\ref{thm:ellipticsolitons}. 

We consider the constraints \eqref{eq:constraint} with $N=3$ at some fixed time; by imposing \eqref{eq:cjconstraint} as $c_3=-c_1-c_2$, using the $\mathrm{U}(1)\times \mathrm{U}(1)$ invariance of the IMM system \eqref{eq:U1timesU1}, we may assume $\lambda\in \R$ (at this fixed time). The constraints are thus given by \eqref{eq:c1c2c3} (with $\alpha(z)$ defined in \eqref{eq:alphaelliptic}).

To proceed, we write
\begin{equation}
c_{2}=w c_{1} \qquad (w\in \C\setminus\{0\})
\end{equation}
so that $c_{3}=-(1+w)c_{1}$. Then, adding the three equations in \eqref{eq:c1c2c3} leads to 
\begin{align}\label{eq:c1c2c3rw1}
&\ii \big( \alpha(a_{1}-a_{1}^*+\ii\delta)+\alpha(a_{3}-a_{3}^*+\ii\delta)-\alpha(a_{1}-a_{3}^*+\ii\delta)-\alpha(a_{3}-a_{1}^*+\ii\delta)\big) \nonumber \\
&-2\,\im\big(w\big(\alpha(a_{2}-a_{1}^*+\ii\delta)-\alpha(a_{2}-a_{3}^*-\ii\delta)-\alpha(a_{3}-a_{1}^*+\ii\delta)+\alpha(a_{3}-a_{3}^*+\ii\delta)\big)\big)	\nonumber \\
&+\ii |w|^2\big( \alpha(a_{2}-a_{2}^*+\ii\delta)+\alpha(a_{3}-a_{3}^*+\ii\delta)-\alpha(a_{2}-a_{3}^*+\ii\delta)-\alpha(a_{3}-a_{2}^*+\ii\delta)\big)= \frac{3}{|c_{1}|^2}.
\end{align}

We rewrite \eqref{eq:c1c2c3rw1} as
\begin{multline}\label{eq:g13g23}
g_{13}(\re(a_{1}-a_{3}))+|w|^2 g_{23}(\re(a_{2}-a_{3}))\\=\ii \big(\alpha(2\ii\,\im(a_{1,0})+\ii\delta)+|w|^2\alpha(2\ii\,\im(a_{2})+\ii\delta)+(1+|w|^2)\alpha(2\ii\,\im(a_{3})+\ii\delta)\big)-\frac{3}{|c_{1}|^2}	 \\
-2\,\im\big(w\big(\alpha(a_{2}-a_{1}^*+\ii\delta)-\alpha(a_{2}-a_{3}^*-\ii\delta)-\alpha(a_{3}-a_{1}^*+\ii\delta)+\alpha(a_{3}-a_{3}^*+\ii\delta)\big)\big),
\end{multline}
where
\begin{equation}\label{eq:gjk}
g_{jk}(z)\coloneqq \ii\big(\alpha(z+\ii\,\im(a_{j}+a_{k})+\ii\delta)-\alpha(z-\ii\,\im(a_{j}+a_{k})+\ii\delta)\big) \quad (j,k=1,2,3),
\end{equation}
in generalization of \eqref{eq:g}. 

Note that \eqref{eq:g13g23} is independent of $\lambda$. To obtain an second equation independent of $\lambda$, we multiply the first equation in \eqref{eq:c1c2c3} by $w$ and subtract the second equation in \eqref{eq:c1c2c3}, yielding
\begin{align}\label{eq:w}
&-\ii w\big(\alpha(a_{1}-a_{1}^*+\ii\delta)-\alpha(a_{1}-a_{3}^*+\ii\delta)-\alpha(a_{2}-a_{1}^*+\ii\delta)+\alpha(a_{2}-a_{3}^*+\ii\delta)\big) \nonumber\\
&-\ii|w|^2\big(\alpha(a_{1}-a_{2}^*+\ii\delta)-\alpha(a_{1}-a_{3}^*+\ii\delta)-\alpha(a_{2}-a_{2}^*+\ii\delta)+\alpha(a_{2}-a_{3}^*+\ii\delta)\big)   =\frac{1-w}{|c_{1}|^2}. 	
\end{align}

Suppose the parameters $a_1,a_2,a_3$ can be chosen so that \eqref{eq:g13g23} and \eqref{eq:w} can be solved for $|c_1|$ and $w$ and consider the first equation in \eqref{eq:c1c2c3}, which we write as 
\begin{equation}\label{eq:threewavec1}
|c_1|^2\bigg(\frac{\lambda}{c_1^*}-\ii\big(\alpha(a_1-a_1^*+\ii\delta)+w^*\alpha(a_1-a_2^*+\ii\delta)-(1+w^*)\alpha(a_1-a_3^*+\ii\delta)\bigg)+1=0.
\end{equation}
By an appropriate choice of phase for $c_1$, \eqref{eq:threewavec1} can be solved to yield a real value for $\lambda$. While we have been unable to find a general procedure to determine $a_1,a_2,a_3$ such that \eqref{eq:g13g23} and \eqref{eq:w} are consistent equations for $|c_1|$ and $w$, we will discuss the simplified problem arising when the real parts of $a_1,a_2,a_3$ are chosen to be identical and $w$ is assumed to be real. 

Suppose $a_1,a_2,a_3$ are given such that $\re(a_1)=\re(a_2)=\re(a_3)$ and assume $w\in \R\setminus\{0\}$. Then, both \eqref{eq:g13g23} and \eqref{eq:w} are quadratic in $w$ and linear in $1/|c_1|^2$ with real coefficients. These equations are straightforwardly solved for $w$ and $|c_1|$ (when a solution exists), after which $\lambda\in \R$ can be found as described above. This provides admissible initial data for Theorem~\ref{thm:ellipticsolitons}; an example solution corresponding to such initial data is given in Figs.~\ref{fig:3sol_ell_1}--\ref{fig:3solitoncs_ell}. 

\begin{figure}
  \includegraphics[width=\linewidth]{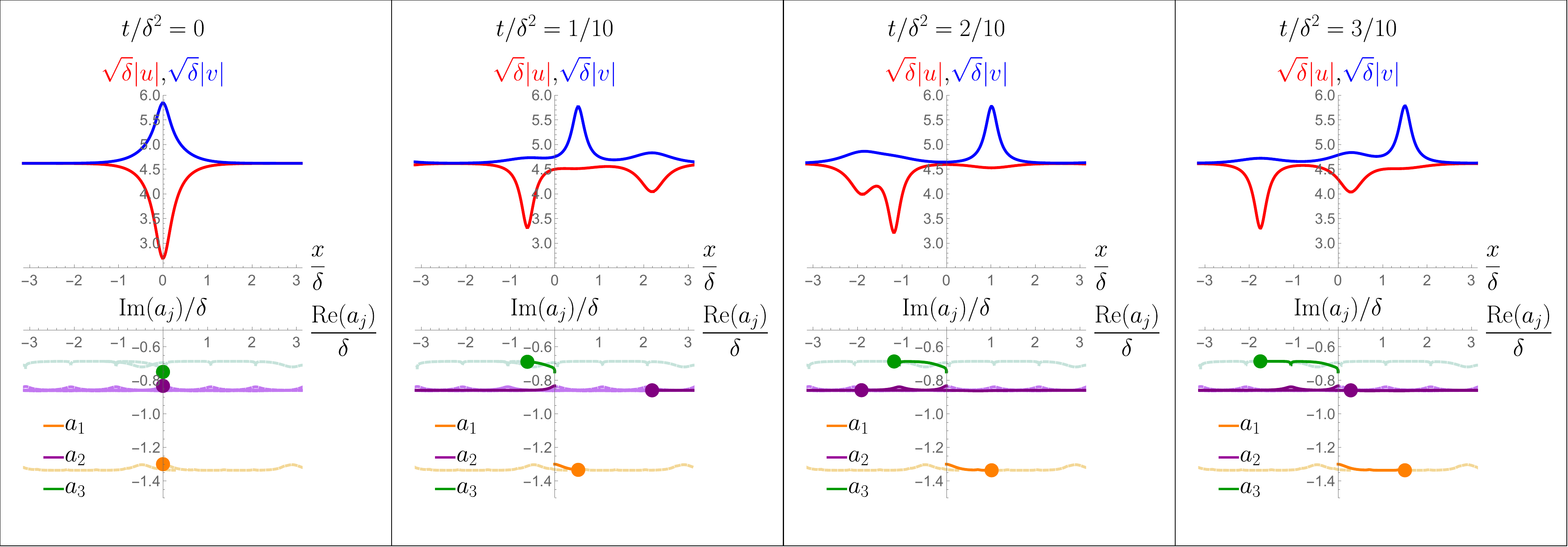}
  \caption{
A periodic three-wave solution with  $\ell= \pi\delta$ and the initial parameter values $\lambda(0) \approx -4.6166/\sqrt{\delta}$, $a_{1}(0)=-13\ii\delta/10$, $a_{2}(0)=-5\ii\delta/6$, $a_{3}(0)=-3\ii\delta/4$, $c_{1}(0)\approx-2.5617\sqrt{\delta}$, $c_{2}(0)\approx 1.0219\sqrt{\delta}$, $c_{3}(0)\approx 1.5398\sqrt{\delta}$ (the initial velocities of $a_1,a_2,a_3$ are determined by \eqref{eq:ajdot}). The resulting dynamics of the parameters is obtained from \eqref{eq:CM} and \eqref{eq:cjdot}. The (amplitudes of the) solution \eqref{eq:ansatz} to the periodic IMM system \eqref{eq:IMM} with \eqref{eq:TTe} is plotted in the upper frames at four points in time. The corresponding dynamics of $a_1,a_2,a_3$ is plotted in the lower frames, where dots indicate the values at each time, the bold lines show the trajectories since $t=0$, and the dotted lines show the future trajectories up to the maximal time $t=13\delta^2/10$.}
  \label{fig:3sol_ell_1}
\end{figure}

\begin{figure}
  \includegraphics[width=\linewidth]{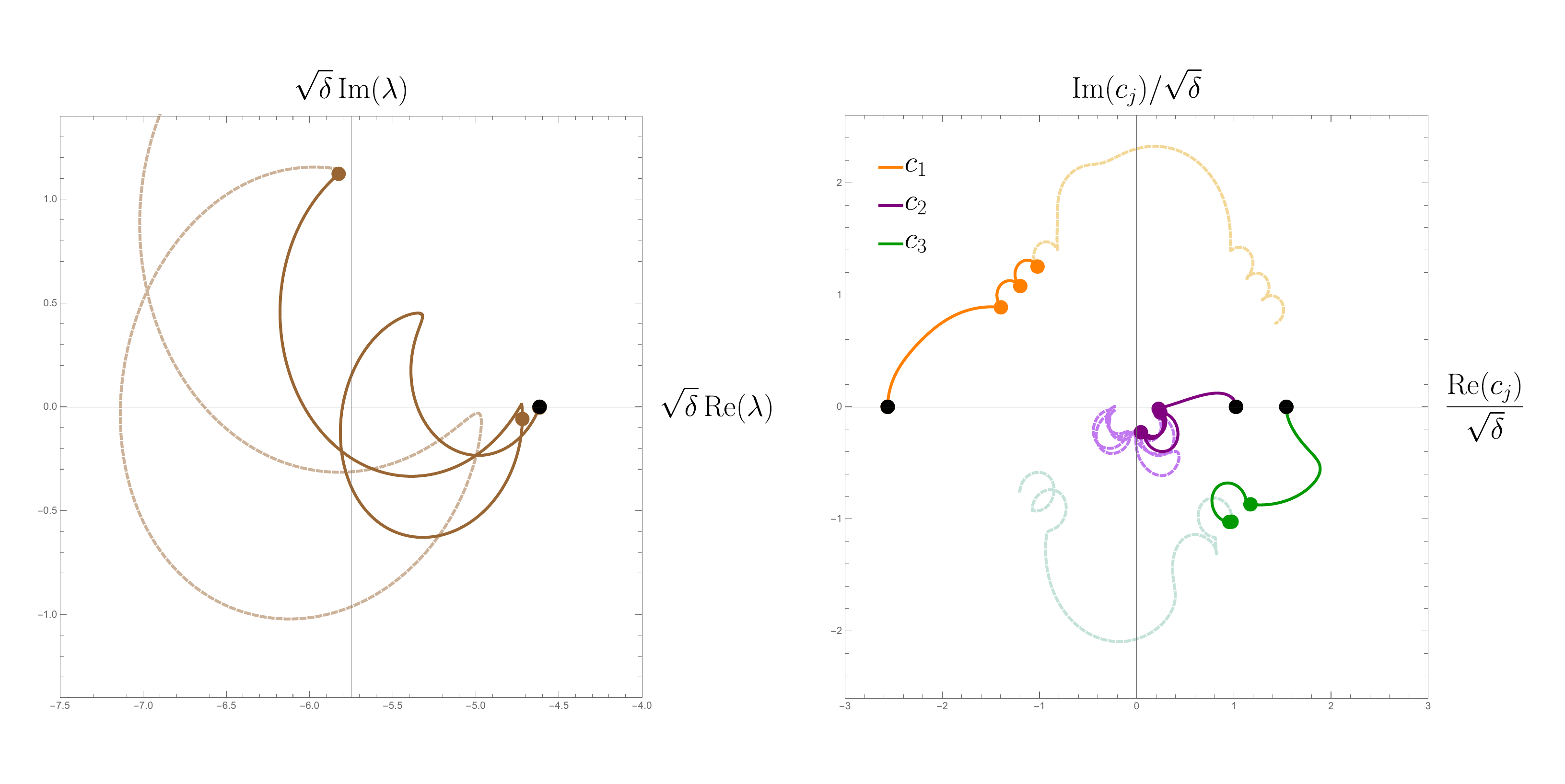}
  \caption{The trajectories of the parameters $\lambda$ and $c_1,c_2,c_3$  in the periodic three-wave solution in Fig.~\ref{fig:3sol_ell_1}. The black dots indicate the positions of $\lambda$ and $c_1,c_2,c_3$ at $t=0$. The colored dots indicate the values at the subsequent times depicted in Fig.~\ref{fig:3sol_ell_1}, the bold lines show the trajectories up to the last time in Fig.~\ref{fig:3sol_ell_1} ($t=3\delta^2/10$), and the dotted lines show the future trajectories, up to the maximal times $t\approx 0.42\delta^2$ for $\lambda$ (where the trajectory exceeds the range of the plot) and $t=\delta^2$ for the $c_1,c_2,c_3$ parameters. Note that for both $c_2$ and $c_3$, the colored dots coincide at certain points in time.}
  \label{fig:3solitoncs_ell}
\end{figure}

\section{Discussion} 

In this paper, we have introduced a new intermediate nonlinear Schr\"{o}dinger system, which interpolates between the mixed Manakov system and the HNLS equation, and solved it by developing relations with CM many-body systems. More specifically, we constructed exact multi-soliton solutions of the IMM system on the real line governed by the hyperbolic CM system and corresponding multi-wave solutions of the periodic IMM system governed by the elliptic CM system. Our results strongly suggest that the IMM system is an integrable model and worthy of further study; we list some such possibilities below.  

\begin{enumerate}

\item We have presented methods to solve the nonlinear constraints of Theorems~\ref{thm:solitons} and \ref{thm:ellipticsolitons} for small values of $N$. Analogous nonlinear constraints have arisen in the study of $N$-soliton solutions of the half-wave maps \cite{berntsonklabbers2020} and spin Benjamin-Ono \cite{berntsonlangmannlenells2022} equations; these constraints can be linearized \cite{berntson2022,berntsonlangmannlenells2022} and systematically solved by linear algebra for arbitrary $N\geq 1$. Thus, it is an interesting question as to whether the same can be done in the present case. We note any such method would expectedly also apply to the soliton solutions of the HNLS equation \cite{matsuno2002}.   

\item Continuum limits of CM systems have been the subject of a number of recent studies \cite{abanov2009,lenzmann2020,ahrend2022,gerard2022}. It would be interesting to understand if there is a relation between the IMM system and continuum limits of the hyperbolic and elliptic CM systems (or the two-particle species generalizations of versions of these CM systems \cite{calogero1976,berntson2020}). 

\item The IMM system is related to two recently-introduced integrable systems: the non-chiral intermediate long wave equation \cite{berntson2020} and the non-chiral intermediate Heisenberg ferromagnet equation \cite{berntson2022}. Like these equations, the IMM system involves the integral operators $T$ and $\tilde{T}$ and admits multi-soliton solutions described by a CM system subject to certain constraints. These equations admit Lax pairs and infinite numbers of conservation laws, so it would be interesting to investigate the existence of similar structures for the IMM system. 

\item The original derivation of the INLS equation \eqref{eq:INLS} as carried out in \cite{pelinovsky1995} uses a multiscale expansion technique, with the intermediate long wave equation as a starting point. It would be interesting to know if there is a similar correspondence between the non-chiral intermediate long wave equation \cite{berntson2020} and the IMM system, i.e., if the IMM system describes envelope waves in the non-chiral intermediate long wave equation.

\item By using the method in Appendix~\ref{app:local}, it may be shown that the $\sigma=+1$ and $\sigma=-1$ cases of the system
\begin{equation}
\begin{split}
\ii u_t=&\; u_{xx}+u(\ii-\sigma T)(|u|^2)_x-\sigma  u\tilde{T}(|v|^2)_x, \\
\ii v_t=&\; v_{xx}+v(\ii-\sigma T)(|v|^2)_x-\sigma  v\tilde{T}(|u|^2)_x
\end{split}
\end{equation}
are reducible to the focusing and defocusing Manakov systems \eqref{eq:manakov}, respectively, in the limit $\delta\downarrow 0$. It would be interesting to study these intermediate versions of the focusing and defocusing Manakov systems from the perspective of integrability.
\end{enumerate}

\paragraph{Acknowledgements.} 
We thank Christopher Ekman, Rob Klabbers, Edwin Langmann, and Jonatan Lenells for useful discussions and collaboration on closely related subjects. We are grateful for stimulating conversations with Katia Gallo and Daniel Qvarng\aa rd on vector nonlinear Schr\"{o}dinger equations and Anton Ottosson on related soliton equations. B.K.B. is supported by the Olle
Engkvist Byggm\"{a}stare Foundation, Grant 211-0122. This article is based on results obtained in the master's thesis of A.F. in theoretical physics at KTH Royal Institute of Technology. 

\appendix

\section{Special functions}\label{app:functional}

We collect identities for the special functions needed in the main text. Identities for elliptic variants of $\alpha(z)$ and $V(z)$, defined in \eqref{eq:alphaelliptic} and Case IV of \eqref{eq:V}, respectively, are given in Appendix~\ref{subsec:elliptic}. The identities for the hyperbolic variants of $\alpha(z)$ and $V(z)$, defined in \eqref{eq:alpha} and Case III of \eqref{eq:V}, respectively, can be obtained by degeneration of the corresponding elliptic identities in the limit $\ell\to\infty$. The procedure for doing so is provided in Appendix~\ref{subsec:hyperbolic}.

\subsection{Elliptic functions}\label{subsec:elliptic}

We refer to \cite[Chapter~23]{DLMF} for definitions of the Weierstrass functions $\zeta(z;\ell,\ii\delta)$ and $\wp(z;\ell,\ii\delta)$. The modifications of these functions we use, $\alpha(z)=\zeta_2(z;\ell,\ii\delta)$ and $V(z)=\wp_2(z;\ell,\ii\delta)$, are defined in terms of these basic functions in \eqref{eq:zeta2}--\eqref{eq:wp2}; the function $\ftwo(z)$ is defined in \eqref{eq:f2}. These functions satisfy the identities 
\begin{align}
\alpha'(z)=&\; -V(z), \label{eq:alphatoV} \\
\alpha(z)^2=&\; V(z)+\ftwo(z), \label{eq:IdV}\\
\alpha(z-a)\alpha(z-b)=&\;\alpha(a-b)\big(\alpha(z-a)-\alpha(z-b)\big) \nonumber\\
&\;+\frac12(\ftwo(z-a)+\ftwo(z-b)+\ftwo(a-b)\big)+\frac{3\zeta(\ii\delta;\ell,\ii\delta)}{2\ii\delta}. \label{eq:Idmain} 
\end{align}
for $z,a,b\in\C$. Moreover, the following periodicity properties hold,
\begin{equation}\label{eq:realperiod}
\alpha(z\pm 2\ell)=\alpha(z)\pm\frac{\pi}{\delta},\qquad V(z\pm 2\ell)=V(z), \qquad \ftwo(z\pm 2\ell)=\ftwo(z)\pm \frac{2\pi}{\delta}\alpha(z)+\bigg(\frac{\pi}{\delta}\bigg)^2
\end{equation}
and
\begin{equation}\label{eq:imperiod}
\alpha(z\pm 2\ii\delta)=\alpha(z),\qquad V(z\pm 2\ii\delta)=V(z),\qquad \ftwo(z\pm 2\ii\delta)=\ftwo(z).
\end{equation}
Proofs of each identity \eqref{eq:IdV}--\eqref{eq:imperiod}, excepting the periodicity properties of $\ftwo(z)$, can be found in \cite[Appendix~A]{berntsonlangmann2021}. The periodicity properties of $\ftwo(z)$ follow from those of $\alpha(z)$ and $V(z)$ in \eqref{eq:realperiod}--\eqref{eq:imperiod} and the definition of $\ftwo(z)$ \eqref{eq:f2}.

The following parity properties hold as consequences of the fact that $\zeta(z)$ is an odd function, $\zeta(-z)=-\zeta(z)$, and $\wp(z)$ is an even function, $\wp(-z)=\wp(z)$, and the definitions \eqref{eq:zeta2}--\eqref{eq:wp2} and \eqref{eq:f2},
\begin{equation}\label{eq:parity}
\alpha(-z)=-\alpha(z),\qquad V(-z)=V(z),\qquad \ftwo(-z)=\ftwo(z), \qquad \ftwo'(-z)=-\ftwo'(z).
\end{equation}

The Weierstrass $\zeta$- and $\wp$-functions with one real and one imaginary half-period are invariant under Schwarz conjugation, $\zeta(z^*)^*=\zeta(z)$ and $\wp(z^*)^*=\wp(z)$. From the definitions \eqref{eq:zeta2}--\eqref{eq:wp2} and \eqref{eq:f2}, it follows that also
\begin{equation}\label{eq:Schwarz}
\alpha(z^*)^*=\alpha(z),\qquad V(z^*)^*=V(z),\qquad \ftwo(z^*)^*=\ftwo(z). 
\end{equation}

\subsection{Hyperbolic functions}\label{subsec:hyperbolic}

Consider the limit
\begin{equation}\label{eq:zetalimit}
\lim_{\ell\to\infty}\zeta(z;\ell,\ii\delta)=\frac{\pi}{2\delta}\coth\bigg(\frac{\pi}{2\delta} z\bigg)-\frac13\bigg(\frac{\pi}{2\delta}\bigg)^2z,
\end{equation}
which gives, in particular,
\begin{equation}\label{eq:zetalimit2}
\lim_{\ell\to\infty} \frac{\zeta(\ii\delta;\ell,\ii\delta)}{\ii\delta}=-\frac13\bigg(\frac{\pi}{2\delta}\bigg)^2.
\end{equation}
It follows from \eqref{eq:zetalimit}-\eqref{eq:zetalimit2} and \eqref{eq:zeta2} that the special functions $\alpha(z)$ and $V(z)$ in the elliptic and hyperbolic cases are related via the limit $\ell\to\infty$,
\begin{equation}\label{eq:alphaVlimits}
\lim_{\ell\to\infty} \zeta_2(z;\ell,\ii\delta)=\frac{\pi}{2\delta}\coth\bigg(\frac{\pi}{2\delta} z\bigg),\qquad \lim_{\ell\to\infty} \wp_2(z;\ell,\ii\delta)=\frac{\big(\frac{\pi}{2\delta}\big)^2}{\sinh^2\big(\frac{\pi}{2\delta} z\big)},
\end{equation}
and, moreover,
\begin{equation}\label{eq:ftwolimit}
\lim_{\ell\to\infty} \ftwo(z;\ell,\ii\delta)=\bigg(\frac{\pi}{2\delta}\bigg)^2,
\end{equation}
using \eqref{eq:f2}. The hyperbolic counterpart of each identity in Appendix~\ref{subsec:elliptic}, excepting \eqref{eq:realperiod}, can be obtained by making the replacements \eqref{eq:zetalimit2}--\eqref{eq:ftwolimit}.

\section{Proofs}\label{app:proofs}

This section contains a detailed proof of Theorem~\ref{thm:ellipticsolitons}. The proof, which is built on three propositions, is given in Appendix~\ref{app:thmproof}. The proofs of the supporting propositions are given in Appendices~\ref{app:ansatzproof}--\ref{app:CMproof}.

\subsection{Proof of Theorem~\ref{thm:ellipticsolitons}}\label{app:thmproof}

We first establish conditions under which the ansatz \eqref{eq:ansatzelliptic} solves the periodic IMM system.

\begin{proposition}\label{prop:ansatz}
The ansatz \eqref{eq:ansatzelliptic} solves the periodic IMM system, \eqref{eq:IMM} with \eqref{eq:TTe}, provided that the time evolution equations \eqref{eq:cjdot}--\eqref{eq:ajdot} and \eqref{eq:lambdadot}, the equality constraints \eqref{eq:cjconstraint} and \eqref{eq:constraintperiodic}, and the inequality constraints \eqref{eq:imaj} and \eqref{eq:ajakelliptic} are satisfied.
\end{proposition}

\begin{proof}
See Appendix~\ref{app:ansatzproof}.	
\end{proof}

We next show that the time evolution equations and two of the constraints in Proposition~\ref{prop:ansatz} are compatible. We first prove this for the constraint \eqref{eq:cjconstraint}.

\begin{lemma}\label{lem:cjconstraint}
Let $\{a_j,c_j\}_{j=1}^N$ be a solution of \eqref{eq:cjdot} on $[0,\tau)$ such that \eqref{eq:cjconstraint} holds at $t=0$ and \eqref{eq:ajakelliptic} holds on $[0,\tau)$. Then, \eqref{eq:cjconstraint} holds on $[0,\tau)$. 
\end{lemma}

\begin{proof}
Let 
\begin{equation}\label{eq:C}
C\coloneqq \sum_{j=1}^N c_j. 	
\end{equation}
By differentiating \eqref{eq:C} with respect to time and inserting \eqref{eq:cjdot}, we have
\begin{align}
\dot{C}=\sum_{j=1}^N \dot{c}_j=2\ii\sum_{j=1}^N\sum_{k\neq j}^N (c_j-c_k)V(a_j-a_k).
\end{align}
As the summand of the double sum is well-defined by the assumption that \eqref{eq:ajakelliptic} holds and anti-symmetric with respect to the interchange $j\leftrightarrow k$ (because $V(z)$ is an even function \eqref{eq:parity}), the double sum vanishes and we have $\dot{C}=0$. Thus, $C$ is equal to its initial value on $[0,\tau)$. In particular, if \eqref{eq:cjconstraint} holds at $t=0$, we have $C=0$ at $t=0$ and consequently on $[0,\tau)$. 
\end{proof}

The preceding lemma is used to prove the following proposition, which states that the constraint \eqref{eq:constraintperiodic} is also compatible with the system of first-order ODEs in Proposition~\ref{prop:ansatz}.

\begin{proposition}\label{prop:constraints}
 Let $\lambda$ and $\{a_j,c_j\}_{j=1}^N$ be a solution of \eqref{eq:lambdadot} and \eqref{eq:cjdot}--\eqref{eq:ajdot} on $[0,\tau)$ such that \eqref{eq:cjconstraint} and \eqref{eq:constraintperiodic} hold at $t=0$ and \eqref{eq:ajakelliptic} holds on $[0,\tau)$. Then, \eqref{eq:constraintperiodic} holds on $[0,\tau)$.  
\end{proposition}

\begin{proof}
See Appendix~\ref{app:constraintproof}. 	
\end{proof}

Together, Propositions~\ref{prop:ansatz}--\ref{prop:constraints} show that a solution of the first-order system of equations \eqref{eq:cjdot}--\eqref{eq:ajdot} and \eqref{eq:lambdadot} on $[0,\tau)$ where (i) the initial conditions satisfy \eqref{eq:cjconstraint} and \eqref{eq:constraintperiodic} and (ii) the conditions \eqref{eq:imaj} and \eqref{eq:ajakelliptic} hold on $[0,\tau)$ can be used to construct a solution of the periodic IMM system. The next proposition states that under certain conditions, a solution of the system of equations consisting of the CM equations of motion \eqref{eq:CM} and the time evolution equations for $\{c_j\}_{j=1}^N$ \eqref{eq:cjdot} and $\lambda$ \eqref{eq:lambdadot} also solves the first-order system of Propositions~\ref{prop:ansatz}--\ref{prop:constraints}. 

\begin{proposition}\label{prop:CM}
Let $\lambda$ and $\{a_j,c_j\}_{j=1}^N$ be a solution of \eqref{eq:CM}, \eqref{eq:cjdot}, and \eqref{eq:lambdadot} on $[0,\tau)$ such that \eqref{eq:ajdot} and \eqref{eq:cjconstraint} hold at $t=0$ and \eqref{eq:ajakelliptic} holds on $[0,\tau)$. Then, \eqref{eq:ajdot} is satisfied on $[0,\tau)$. 
\end{proposition}

\begin{proof}
See Appendix~\ref{app:CMproof}.	
\end{proof}

The observations above can now be used to prove the theorem. Suppose we are given $\lambda$ and $\{a_j,c_j\}_{j=1}^N$ that satisfy (i) the equations of motion \eqref{eq:CM}, \eqref{eq:cjdot}, and \eqref{eq:lambdadot} on $[0,\tau)$, (ii) the equality constraints \eqref{eq:cjconstraint} and \eqref{eq:constraintperiodic} at $t=0$, and (iii) the inequality constraints \eqref{eq:imaj} and \eqref{eq:ajakelliptic} on $[0,\tau)$. By Proposition~\ref{prop:CM}, we have that the first-order equations of motion for $\{a_j\}_{j=1}^N$ \eqref{eq:ajdot} are satisfied on $[0,\tau)$. Lemma~\ref{lem:cjconstraint} and Proposition~\ref{prop:constraints} then show that the constraints \eqref{eq:cjconstraint} and \eqref{eq:constraintperiodic}, respectively, hold on $[0,\tau)$. The assumptions of Proposition~\ref{prop:ansatz} are now satisfied and consequently the ansatz \eqref{eq:ansatzelliptic} with our supposed solution $\lambda$ and $\{a_j,c_j\}_{j=1}^N$ provides a solution of the periodic IMM system on $[0,\tau)$. This completes the proof.

\subsection{Proof of Proposition~\ref{prop:ansatz}}\label{app:ansatzproof}

Using \eqref{eq:IMM2}, the ansatz \eqref{eq:ansatzelliptic} may be written as
\begin{equation}\label{eq:ansatz2}
U= \ee^{2\ii N \cc t} \Bigg(\lambda E + \ii \sum_{j=1}^N c_j A_+(x-a_j)\Bigg)
\end{equation}
with 
\begin{equation}\label{eq:AE}
E\coloneqq \left(\begin{array}{c} 1 \\ -1 \end{array}\right)
\end{equation}
and $A_{\pm}(z)$ as in \eqref{eq:Apm} (the function $A_-$ appears in $U^*$ and is thus needed below). Observe that $E$ in \eqref{eq:AE} acts as the identity on two-vectors under the operation $\circ$ defined in \eqref{eq:circ}.

We compute each term in \eqref{eq:IMM2} with \eqref{eq:ansatz2} and \eqref{eq:cjconstraint}. The first two terms are found to be
\begin{equation}\label{eq:Ut}
\ii U_t = - 2N \cc U+\ii\ee^{2\ii N\cc t}\Bigg( \dot{\lambda} E+\ii \sum_{j=1}^N \big(\dot{c}_j A_+(x-a_j)-c_j\dot{a}_jA_+'(x-a_j)\big)\Bigg)
\end{equation}
and 
\begin{equation}\label{eq:Uxx}
U_{xx}=\ii \ee^{2\ii N \cc t}\sum_{j=1}^N c_j A_+''(x-a_j).
\end{equation}
We compute the nonlinear term in \eqref{eq:IMM2} in several steps starting from
\begin{equation}\label{eq:UUx1}
(U\circ U^*)_x= -\ii \lambda \sum_{j=1}^N c_j^* A_-'(x-a_j^*)+\ii\lambda^* \sum_{j=1}^N c_j A_+'(x-a_j)   +\sum_{j=1}^N\sum_{k=1}^N c_jc_k^* \big(A_+(x-a_j)\circ A_-(x-a_k^*)\big)_x,
\end{equation}
where we have used that $A_+(x-a_k)^*=A_-(x-a_k^*)$ as a consequence of \eqref{eq:Schwarz}. To proceed, we need the identity
\begin{align}\label{eq:solId1}
A_+(x-a_j)\circ A_-(x-a_k^*)=&\; \alpha(a_j-a_k^*+\ii\delta)\big(A_+(x-a_j)-A_-(x-a_k^*)\big) \nonumber \\
&\; +\frac12 \big(K_+(x-a_j)+K_-(x-a_k^*)+\ftwo(a_j-a_k^*+\ii\delta)E\big)+\frac{3\zeta(\ii\delta)}{2\ii\delta}E,
\end{align}
where 
\begin{equation}\label{eq:Kpm}
K_{\pm}(z)\coloneqq \left(\begin{array}{c}
\ftwo(z\mp \ii\delta/2) \\
-\ftwo(z\pm\ii\delta/2)
\end{array}\right),
\end{equation}
which follows from \eqref{eq:Idmain} by specializing variables. By differentiating \eqref{eq:solId1} with respect to $x$, we find
\begin{equation}
\big(A_+(x-a_j)\circ A_-(x-a_k^*)\big)_x=\alpha(a_j-a_k^*+\ii\delta)\big(A_+'(x-a_j)-A_-'(x-a_k^*)\big)+\frac12\big(K_+'(x-a_j)+K_-'(x-a_k^*)\big)
\end{equation}
and inserting this into \eqref{eq:UUx1} gives
\begin{align}\label{eq:UUx2}
(U\circ U^*)_x=&\; -\ii \lambda \sum_{j=1}^N c_j^* A_-'(x-a_j^*) + \ii\lambda^* \sum_{j=1}^N c_j A_+'(x-a_j)  \nonumber \\
&\;  +\sum_{j=1}^N\sum_{k=1}^N c_jc_k^* \alpha(a_j-a_k^*+\ii\delta)\big(A_+'(x-a_j)-A_-'(x-a_k^*)\big) \nonumber \\
&\; +\frac12\sum_{j=1}^N\sum_{k=1}^N  c_jc_k^* \big(K_+'(x-a_j)+K_-'(x-a_k^*)\big),
\end{align}
but the final line vanishes by \eqref{eq:cjconstraint}.

We now require the following result from \cite{berntson2022elliptic}, which will be used to compute the action of $\cT$ on $(U\circ U^*)_x$.

\begin{lemma}\label{lem:cT}
When \eqref{eq:imaj} is satisfied, the operator $\cT$ has the actions \eqref{eq:TA} on the functions $A_{+}'(x-a_j)$ and $A_{-}'(x-a_j^*)$.
\end{lemma}

It thus follows from \eqref{eq:TA} that
\begin{equation}\label{eq:PiApm}
((\ii+ \cT) A'_{+}(\cdot-a_j))(x)=2\ii A_{+}'(x-a_j)+ 2\ii\cc\left(\begin{array}{c} 0 \\ 1 \end{array}\right),\qquad ((\ii+ \cT) A'_{-}(\cdot-a_j^*))(x)= 2\ii\cc\left(\begin{array}{c} 1 \\ 0 \end{array}\right)
\end{equation}
under the condition \eqref{eq:imaj}. Hence, applying $(\ii+\cT)$ to \eqref{eq:UUx2} (with the final line removed) gives
\begin{align}\label{eq:PiUUx}
(\ii+\cT)(U\circ U^*)_x = &\; -2 \lambda^*\sum_{j=1}^N c_j A_+'(x-a_j)+ \ii\sum_{j=1}^N\sum_{k=1}^N c_jc_k^* \alpha(a_j-a_k^*+\ii\delta)\big(2A_+'(x-a_j)-2\cc E) \nonumber \\
=&\; -2\sum_{j=1}^N c_j A_+'(x-a_j)\Bigg(\lambda^*-\ii\sum_{k=1}^N c_k^*\alpha(a_j-a_k^*+\ii\delta)\Bigg)-2\ii\cc \sum_{j=1}^N\sum_{k=1}^N c_jc_k^*\alpha(a_j-a_k^*+\ii\delta)E. 
\end{align}
By inserting \eqref{eq:constraint} and 
\begin{equation}
\sum_{j=1}^N\sum_{k=1}^N c_jc_k^*\alpha(a_j-a_k^*+\ii\delta)=-\ii N,	
\end{equation}
which follows from \eqref{eq:constraint} and \eqref{eq:cjconstraint}, into the first and second terms in \eqref{eq:PiUUx}, respectively, we obtain
\begin{equation}
\label{eq:PiUUx2}
(\ii+\cT)(U\circ U)_x= -2 N\cc E+2 \sum_{j=1}^N A_+'(x-a_j). 
\end{equation}

Next, $\circ$-multiplication of \eqref{eq:PiUUx2} by $U$ leads to 
\begin{equation}\label{eq:PiUUx3}
U\circ (\ii+\cT)(U\circ U)_x =   -2N\cc U  +2\ee^{2\ii N\cc t}\lambda\sum_{j=1}^N A_+'(x-a_j)  +2\ii \ee^{2\ii N\cc t} \sum_{j=1}^N \sum_{k=1}^N c_jA_+(x-a_j)\circ A_+'(x-a_k).
\end{equation}

We will now employ the identities 
\begin{equation}\label{eq:solId2}
A_+(x-a_j)\circ A_+'(x-a_j)=-\frac12 A_+''(x-a_j)+\frac12 K_+'(x-a_j)
\end{equation}
and
\begin{align}\label{eq:solId3}
A_+(x-a_j)\circ A_+'(x-a_k)=&\; -\alpha(a_j-a_k)A_+'(x-a_k)-V(a_j-a_k)\big(A_+(x-a_j)-A_+(x-a_k)\big) \nonumber \\
&\; +\frac12 K_+'(x-a_k)+\frac12\ftwo'(a_j-a_k)E,
\end{align}
which are obtained by differentiating \eqref{eq:IdV} and \eqref{eq:Idmain}, respectively and specializing variables. We evaluate the double sum in \eqref{eq:PiUUx3}, using \eqref{eq:solId2} and \eqref{eq:solId3} for terms with $j=k$ and $j\neq k$, respectively; this leads to 
\begin{align}\label{eq:PiUUx4}
U\circ (\ii+\cT)(U\circ U)_x=&\; -2N\cc U +2\ee^{2\ii N\cc t}\lambda\sum_{j=1}^N A_+'(x-a_j) -\ii \ee^{2\ii N\cc t} \sum_{j=1}^N c_j\big(A_+''(x-a_j)-K_+'(x-a_j)\big) \nonumber \\
&\; -2\ii \ee^{2\ii N \cc t} \sum_{j=1}^N \sum_{k\neq j}^N c_j \alpha(a_j-a_k)A_+'(x-a_k) \nonumber \\
&\; -2\ii \ee^{2\ii N\cc t}\sum_{j=1}^N \sum_{k\neq j}^N c_j V(a_j-a_k)\big(A_+(x-a_j)-A_+(x-a_k)\big) \nonumber \\
&\; +\ii \ee^{2\ii N\cc t} \sum_{j=1}^N\sum_{k\neq j}^N c_j K_+'(x-a_k)   +\ii \ee^{2\ii N\cc t} \sum_{j=1}^N\sum_{k\neq j}^N c_j\ftwo'(a_j-a_k)E.
\end{align}
All terms in $K_+'(x-a_j)$ in \eqref{eq:PiUUx4} vanish by \eqref{eq:cjconstraint} (to see this, swap $j\leftrightarrow k$ in the first double sum in the fourth line and add to this the $K_+'(x-a_j)$ terms in the first line). Using this, swapping $j\leftrightarrow k$ in the double sum in the second line, rewriting the double sum in the third line as
\begin{multline}\label{eq:UPiUUx4}
\sum_{j=1}^N\sum_{k\neq j}^N  c_j V(a_j-a_k)\big(A_+(x-a_j)-A_+(x-a_k)\big) \\
=\frac12 \sum_{j=1}^N \sum_{k\neq j}^N (c_j-c_k) V(a_j-a_k)\big(A_+(x-a_j)-A_+(x-a_k)\big) \\
=\sum_{j=1}^N \sum_{k\neq j}^N (c_j-c_k)V(a_j-a_k)A_+(x-a_j),
\end{multline}
symmetrizing the second double sum in the fourth line, and rearranging, \eqref{eq:PiUUx4} becomes
\begin{align}\label{eq:PiUUx5}
U\circ (\ii+\cT)(U\circ U)_x=&\; -2N\cc U +\frac{\ii}{2} \ee^{2\ii N\cc t} \sum_{j=1}^N\sum_{k\neq j}^N (c_j-c_k)\ftwo'(a_j-a_k)E \nonumber\\
&\; -2\ii \ee^{2\ii N \cc t}\sum_{j=1}^N\sum_{k\neq j}^N (c_j-c_k) V(a_j-a_k)A_+(x-a_j)  \nonumber \\
&\; +\ee^{2\ii N\cc t}\sum_{j=1}^N A_+'(x-a_j)\Bigg( 2\lambda-2\ii \sum_{k\neq j}^N c_k \alpha(a_j-a_k)\Bigg)-\ii \ee^{2\ii N\cc t}\sum_{j=1}^N c_j A_+''(x-a_j).
\end{align}

Inserting \eqref{eq:Ut}, \eqref{eq:Uxx}, and \eqref{eq:PiUUx5} into \eqref{eq:IMM2} yields
\begin{align}\label{eq:IMM3}
\ii U_t-U_{xx}- U\circ (\ii+\cT)(U\circ U^*)_x=&\;  \ee^{2\ii N\cc t}\Bigg(\ii\dot{\lambda}+\frac{\ii}{2}\sum_{j=1}^N\sum_{k\neq j}^N (c_j-c_k) \ftwo'(a_j-a_k)\Bigg)E \nonumber \\
&\;     +\ee^{2\ii N\cc t}\sum_{j=1}^N A_+(x-a_j) \Bigg(-\dot{c}_j+2\ii\sum_{k\neq j}^N (c_j-c_k)V(a_j-a_k)\Bigg) \nonumber \\
&\;  +\ee^{2\ii N\cc t}\sum_{j=1}^N A_+'(x-a_j)\Bigg(c_j\dot{a}_j-2\lambda-2\ii\sum_{k\neq j}^N c_k \alpha(a_j-a_k) \Bigg).
\end{align}
Then, the linear independence of $E$ and $\{A_+(x-a_j),A_+'(x-a_j)\}_{j=1}^N$ for distinct $\{a_j\}_{j=1}^N$ \eqref{eq:ajakelliptic} implies that the periodic IMM system \eqref{eq:IMM2} is satisfied when \eqref{eq:cjdot}--\eqref{eq:ajdot} and \eqref{eq:lambdadot} hold.

\subsection{Proof of Propositon~\ref{prop:constraints}}\label{app:constraintproof} 

We start by making the definition
\begin{equation}\label{eq:Dj}
D_j\coloneqq c_j\bigg(\lambda^*-\ii \sum_{k=1}^N c_k^*\alpha(a_j-a_k^*+\ii\delta)\bigg) \quad (j=1,\ldots,N).
\end{equation}
We will show each $D_j$ is conserved when the initial conditions $\{D_j(0)=-1\}_{j=1}^N$ are satisfied. In this proof we repeatedly use that $\alpha(z)$ is a $2\ii\delta$-periodic, odd function and that $V(z)$ is a $2\ii\delta$-periodic, even function; see \eqref{eq:imperiod}--\eqref{eq:parity}.

By differentiating \eqref{eq:Dj} with respect to $t$ using and rearranging, we have
\begin{align}\label{eq:Djdotmaster}
\dot{D}_j-c_j\dot{\lambda}^*=&\;  \dot{c}_j\Bigg(\lambda^*-\ii\sum_{k=1}^N c_k^*\alpha(a_j-a_k^*+\ii\delta)\Bigg)+\ii\sum_{k=1}^Nc_jc_k^*\dot{a}_jV(a_j-a_k^*+\ii\delta)\nonumber\\
&\;-\ii \sum_{k=1}^N c_j \dot{c}_k^*\alpha(a_j-a_k^*+\ii\delta) 
-\ii\sum_{k=1}^N c_jc_k^*\dot{a}_k^*V(a_j-a_k^*+\ii\delta),
\end{align}
We consider the first and second lines of \eqref{eq:Djdotmaster} separately.

By inserting \eqref{eq:cjdot} and \eqref{eq:ajdot}, the first line of \eqref{eq:Djdotmaster} becomes 
\begin{align}\label{eq:Djdot2line1}
&2\ii \lambda^* \sum_{k\neq j}^N (c_j-c_k)V(a_j-a_k)+2\sum_{k=1}^N\sum_{	l\neq j}^N c_k^*(c_j-c_l)\alpha(a_j-a_k^*+\ii\delta)V(a_j-a_l) \nonumber \\
&+2\ii\lambda \sum_{k=1}^N c_k^*V(a_j-a_k^*+\ii\delta)-2\sum_{k=1}^N\sum_{l\neq j}^N c_k^*c_l\alpha(a_j-a_l)V(a_j-a_k^*+\ii\delta)
\end{align}
To proceed, we need the following identity
\begin{multline}\label{eq:Djdotid1}
\alpha(a_j-a_k^*+\ii\delta)V(a_j-a_l)+\alpha(a_j-a_l)V(a_j-a_k^*+\ii\delta)= \\ \alpha(a_k^*-a_l+\ii\delta)\big(V(a_j-a_k^*+\ii\delta)-V(a_j-a_l)\big)-\frac12\big(\ftwo'(a_j-a_k^*+\ii\delta)+\ftwo'(a_j-a_l)\big),
\end{multline}
which is obtained by differentiating \eqref{eq:Idmain} and specializing variables. Inserting \eqref{eq:Djdotid1} into \eqref{eq:Djdot2line1} gives
\begin{align}\label{eq:Djdot2line1step2}
&2\ii \lambda^* \sum_{k\neq j}^N (c_j-c_k)V(a_j-a_k)+2\ii\lambda\sum_{k=1}^N c_k^*V(a_j-a_k^*+\ii\delta)	 \nonumber \\
&+2\sum_{k\neq j}^N\sum_{l=1}^N c_jc_l^*\alpha(a_j-a_l^*+\ii\delta)V(a_j-a_k)-2\sum_{k\neq j}^N\sum_{l=1}^N c_kc_l^* \alpha(a_k-a_l^*+\ii\delta)V(a_j-a_l) \nonumber\\
&-2\sum_{k=1}^N\sum_{l\neq j}^N c_k^*c_l \alpha(a_k^*-a_l+\ii\delta)V(a_j-a_k^*+\ii\delta) +\sum_{k=1}^N\sum_{l\neq j}^N c_k^*c_l \big(\ftwo'(a_j-a_k^*+\ii\delta)+\ftwo'(a_j-a_l)\big),
\end{align}
where we have renamed indices $k\leftrightarrow l$ in the second line. The definition of $D_j$ \eqref{eq:Dj} allows us to rewrite \eqref{eq:Djdot2line1step2} as
\begin{align}\label{eq:Djdot2line1step3}
&2\ii\sum_{k\neq j}^N (D_j-D_k)V(a_j-a_k)+2\ii\lambda \sum_{k=1}^N c_k^* V(a_j-a_k^*+\ii\delta)\nonumber\\
&-2\sum_{k=1}^N\sum_{l\neq j}^N c_k^*c_l \alpha(a_k^*-a_l+\ii\delta)V(a_j-a_k^*+\ii\delta) +\sum_{k=1}^N\sum_{l\neq j}^N c_k^*c_l \big(\ftwo'(a_j-a_k^*+\ii\delta)+\ftwo'(a_j-a_l)\big).
\end{align}

We now turn to the second line of \eqref{eq:Djdotmaster}, which upon insertion of \eqref{eq:cjdot} and \eqref{eq:ajdot} reads
\begin{equation}\label{eq:Djdotline2}
-2\sum_{k=1}^N\sum_{l\neq k}^N c_j(c_k^*-c_l^*)\alpha(a_j-a_k^*+\ii\delta)V(a_k^*-a_l^*)-2\ii\sum_{k=1}^Nc_j\Bigg(\lambda^*-\ii\sum_{l\neq k}^Nc_l^*\alpha(a_k^*-a_l^*)\Bigg)V(a_j-a_k^*+\ii\delta).
\end{equation}
By making the replacements $a_j\rightarrow a_k^*$, $a_k^*\rightarrow a_j$, $a_l\rightarrow a_l^*$ in \eqref{eq:Djdotid1}, we find
\begin{multline}\label{eq:Djdotid2}
	\alpha(a_j-a_k^*+\ii\delta)V(a_k^*-a_l^*) -\alpha(a_k^*-a_l^*)V(a_j-a_k^*+\ii\delta)= \\  \alpha(a_j-a_l^*-\ii\delta)\big(V(a_k^*-a_l^*)-V(a_j-a_k^*+\ii\delta)\big)
 -\frac{1}{2}\big(\ftwo'(a_j-a_k^*+\ii\delta)-\ftwo'(a_k^*-a_l^*)\big).
\end{multline}
Inserting \eqref{eq:Djdotid2} into \eqref{eq:Djdotline2} and simplifying, we arrive at 
\begin{align}\label{eq:Djdotline2step2}
&-2\ii\lambda^*\sum_{k=1}^N c_j V(a_j-a_k^*+\ii\delta)-2\sum_{k=1}^N\sum_{l\neq k}^N c_j\big(c_k^*\alpha(a_j-a_k^*+\ii\delta)-c_l^*\alpha(a_j-a_l^*+\ii\delta)\big)V(a_k^*-a_l^*)	 \nonumber \\
& -2\sum_{k=1}^N\sum_{l\neq k}^N c_jc_l^* \alpha(a_j-a_l^*+\ii\delta)V(a_j-a_k^*+\ii\delta)-\sum_{k=1}^N\sum_{l\neq k}^N c_jc_l^*\ftwo'(a_j-a_k^*+\ii\delta)+ \sum_{k=1}^N\sum_{l\neq k}^N c_jc_l^*\ftwo'(a_k^*-a_l^*).
\end{align}

It is possible to further simplify \eqref{eq:Djdotline2step2} in two ways. Firstly, we note that the double sum in the first line of \eqref{eq:Djdotline2step2} vanishes by symmetry. Secondly, using \eqref{eq:lambdadot}, the last sum in \eqref{eq:Djdotline2step2} can be reexpressed as
\begin{equation}\label{eq:DjdotSimpl2}
\sum_{k=1}^N \sum_{l\neq k}^N c_jc_l^*\ftwo'(a_k^*-a_l^*) = -\frac{1}{2}\sum_{k=1}^N\sum_{l\neq k}^N c_j(c_k^*-c_l^*)\ftwo'(a_k^*-a_l^*) = -c_j\dot{\lambda}^*.
\end{equation} 
Using these simplifications,  \eqref{eq:Djdotline2step2} becomes
\begin{align}\label{eq:Djdotline2step3}
&-2\ii\lambda^*\sum_{k=1}^N c_j V(a_j-a_k^*+\ii\delta) -2\sum_{k=1}^N\sum_{l\neq k}^N c_jc_l^* \alpha(a_j-a_l^*+\ii\delta)V(a_j-a_k^*+\ii\delta) \nonumber \\
&-\sum_{k=1}^N\sum_{l\neq k}^N c_jc_l^*\ftwo'(a_j-a_k^*+\ii\delta)-c_j\dot{\lambda}^*
\end{align}

Let us return to our original expression. Inserting \eqref{eq:Djdot2line1step3} and \eqref{eq:Djdotline2step3} in place of the first and second lines of \eqref{eq:Djdotmaster}, respectively leads to the equation
\begin{align}\label{eq:Djdotmasterstep2}
\dot{D}_j=&\; 2\ii \sum_{k\neq j}^N (D_j-D_k)V(a_j-a_k) \nonumber \\ 
&\;-2\ii \lambda^* \sum_{k=1}^N c_j V(a_j-a_k^*+\ii\delta)-2\sum_{k=1}^N\sum_{l\neq k}^N c_jc_l^* \alpha(a_j-a_l^*+\ii\delta)V(a_j-a_k^*+\ii\delta) \nonumber \\
&\; +2\ii\lambda\sum_{k=1}^N c_k^*V(a_j-a_k^*+\ii\delta)-2\sum_{k=1}^N\sum_{l\neq j}^N c_k^*c_l \alpha(a_k^*-a_l+\ii\delta)V(a_j-a_k^*+\ii\delta) \nonumber \\
&\; +\sum_{k=1}^N\sum_{l\neq j}^N c_k^*c_l \big(\ftwo'(a_j-a_k^*+\ii\delta)+\ftwo'(a_j-a_l)\big)-\sum_{k=1}^N\sum_{l\neq k}^N c_jc_l^*\ftwo'(a_j-a_k^*+\ii\delta).
\end{align}
We may rewrite the second and third lines of
\eqref{eq:Djdotmasterstep2} as, respectively,
\begin{equation}\label{eq:firstline}
-2\ii \sum_{k=1}^N c_j  \Bigg(\lambda^*-\ii\sum_{l\neq k}^N c_l^*\alpha(a_j-a_l^*+\ii\delta)\Bigg)V(a_j-a_k^*+\ii\delta)=-2\ii\sum_{k=1}^N\big(D_j+\ii c_j\alpha(a_j-a_k^*+\ii\delta)V(a_j-a_k^*+\ii\delta)
\end{equation}
and
\begin{align}\label{eq:secondline}
2\ii\sum_{k=1}^N c_k^*\Bigg(\lambda+\ii\sum_{l\neq j}^N c_l\alpha(a_k^*-a_l+\ii\delta)\Bigg)V(a_j-a_k^*+\ii\delta)=2\sum_{k=1}^N (D_k^*+\ii c_j\alpha(a_j-a_k^*+\ii\delta)   \big)V(a_j-a_k^*+\ii\delta). 
\end{align}
The fourth line of \eqref{eq:Djdotmasterstep2} is now shown to vanish. From Lemma~\ref{lem:cjconstraint} and \eqref{eq:cjconstraint}, we have that $\sum_{l\neq j}^Nc_l=-c_j$, $\sum_{k=1}^N c_k^*=0$, and $\sum_{l\neq k} c_l^*=-c_k^*$; inserting these into the third line in \eqref{eq:Djdotmasterstep2} gives
\begin{align}\label{eq:ftwopsum}
-\sum_{k=1}^N c_jc_k^*\ftwo'(a_j-a_k^*+\ii\delta)+\sum_{k=1}^Nc_jc_k^* \ftwo'(a_j-a_k^*+\ii\delta)=0. 	
\end{align}
Thus, replacing the second, third, and fourth lines of \eqref{eq:Djdotmasterstep2} by \eqref{eq:firstline}, \eqref{eq:secondline}, and \eqref{eq:ftwopsum}, respectively, we obtain the following system of linear ODEs for $\{D_j\}_{j=1}^N$,
\begin{equation}\label{eq:Djdonev2}
\dot{D}_j=2\ii\sum_{k\neq j}^N (D_j-D_k)V(a_j-a_k)-2\ii\sum_{k=1}^N (D_j-D_k^*)V(a_j-a_k^*+\ii\delta)\quad (j=1,\ldots,N). 
\end{equation}

It is clear that \eqref{eq:Djdonev2} admits the solution $\{D_j(t)=-1\}_{j=1}^N$ on $[0,\tau)$. To see that this solution is unique when $\{D_j(0)=-1\}_{j=1}^N$, it suffices to note that all coefficients of $\{D_j,D_j^*\}_{j=1}^N$ in the linear system \eqref{eq:Djdonev2} are regular on $[0,\tau)$ as a consequence of the assumption that \eqref{eq:ajakelliptic} holds on $[0,\tau)$. By comparing the constraint \eqref{eq:constraintperiodic} with the definition of $D_j$ \eqref{eq:Dj}, we understand that this solution is equivalent to the conservation of \eqref{eq:constraintperiodic} on $[0,\tau)$. This completes our proof.

\subsection{Proof of Proposition~\ref{prop:CM}}\label{app:CMproof}

We first establish the following lemma. 

\begin{lemma}\label{lem:cjneq0}
Let $\lambda$ and $\{a_j,c_j\}_{j=1}^N$ be a solution of \eqref{eq:cjdot}--\eqref{eq:ajdot} on $[0,\tau)$ such that \eqref{eq:cjconstraint} and \eqref{eq:constraintperiodic} hold at $t=0$ and \eqref{eq:ajakelliptic} holds on $[0,\tau)$. Then,
\begin{equation}\label{eq:cjneq0}
	c_j\neq 0 \quad (j=1,\ldots,N)
\end{equation}
holds on $[0,\tau)$. 
\end{lemma}

\begin{proof}
The constraint \eqref{eq:constraintperiodic} will be violated if $c_j=0$ for any $j=1,\ldots,N$. Under our assumptions, Proposition~\ref{prop:constraints} guarantees that \eqref{eq:constraintperiodic} holds on $[0,\tau)$, so this would be a contradiction.  
\end{proof}

Using Lemma~\ref{lem:cjneq0}, we relate certain solutions of the systems of equations (i) \eqref{eq:cjdot}--\eqref{eq:ajdot} and \eqref{eq:lambdadot} and (ii) \eqref{eq:CM}, \eqref{eq:cjdot}, and \eqref{eq:lambdadot}.

\begin{lemma}\label{lem:CM}
Let $\lambda$ and $\{a_j,c_j\}_{j=1}^N$ be a solution of \eqref{eq:cjdot}--\eqref{eq:ajdot} and \eqref{eq:lambdadot} on $[0,\tau)$ such that \eqref{eq:cjconstraint} and \eqref{eq:constraintperiodic} hold at $t=0$ and \eqref{eq:ajakelliptic} holds on $[0,\tau)$. Then, \eqref{eq:CM} holds on $[0,\tau)$. 
\end{lemma}

\begin{proof}
By differentiating \eqref{eq:ajdot} with respect to time and inserting \eqref{eq:cjdot}, we compute
\begin{align}\label{eq:BTcalc1} \nonumber 
&c_j\ddot{a}_j-2\dot{\lambda} \\
&=-\dot{c}_j\dot{a}_j+2\ii \sum_{k\neq j}^N \big(\dot{c}_k\alpha(a_j-a_k)-c_k(\dot{a}_j-\dot{a}_k)V(a_j-a_k)\big) \nonumber \\
& = -2\ii \dot{a}_j\sum_{k\neq j}^N (c_j-c_k)V(a_j-a_k)+2\ii \sum_{k\neq j}^N\Bigg(2\ii\sum_{l\neq k}^N (c_k-c_l)\alpha(a_j-a_k)V(a_k-a_l)-c_k(\dot{a}_j-\dot{a}_k)V(a_j-a_k)\Bigg) \nonumber \\
&= -2\ii \sum_{k\neq j}^N (c_j\dot{a}_j-c_k\dot{a}_k)V(a_j-a_k)-4\sum_{k\neq j}^N\sum_{l\neq k}^N (c_k-c_l)\alpha(a_j-a_k)V(a_k-a_l).
\end{align}

Next, by inserting \eqref{eq:ajdot} into \eqref{eq:BTcalc1}, isolating diagonal terms, and simplifying, we find
\begin{align}\label{eq:BTcalc2}
&c_j\ddot{a}_j-2\dot{\lambda} \nonumber \\
&= 4\sum_{k\neq j}^N\Bigg( \sum_{l\neq j}^N c_l\alpha(a_j-a_l)-\sum_{l\neq k}^N c_l\alpha(a_k-a_l)\Bigg)V(a_j-a_k)-4\sum_{k\neq j}^N\sum_{l\neq k}^N (c_k-c_l)\alpha(a_j-a_k)V(a_k-a_l) \nonumber \\
&= 4\sum_{k\neq j}^N\sum_{l\neq j,k}^N c_l\big(\alpha(a_j-a_l)-\alpha(a_k-a_l)\big)V(a_j-a_k)+4\sum_{k\neq j}^N(c_j+c_k)\alpha(a_j-a_k)V(a_j-a_k) \nonumber \\
&\phantom{=} -4\sum_{k\neq j}^N\sum_{l\neq j, k}^N (c_k-c_l)\alpha(a_j-a_k)V(a_k-a_l)+4\sum_{k\neq j}^N (c_j-c_k)\alpha(a_j-a_k)V(a_j-a_k) \nonumber \\
&= 4\sum_{k\neq j}^N\sum_{l\neq j,k}^N c_l\big(\alpha(a_j-a_l)-\alpha(a_k-a_l)\big)V(a_j-a_k) \nonumber \\
&\phantom{=} +4\sum_{k\neq j}^N\sum_{l\neq j, k}^N c_l \big(\alpha(a_j-a_k)-\alpha(a_j-a_l) \big)   V(a_k-a_l)+8\sum_{k\neq j}^N c_j\alpha(a_j-a_k)V(a_j-a_k).
\end{align}

We now insert the identities
\begin{align}
\big(\alpha(a_j-a_l)-\alpha(a_k-a_l)\big)V(a_j-a_k)=&\; -\big(\alpha(a_j-a_k)-\alpha(a_j-a_l)\big)V(a_k-a_l) \nonumber \\
&\; -\frac12\big(\ftwo'(a_j-a_k)-\ftwo'(a_k-a_l)\big),
\end{align}
which can be obtained by differentiating \eqref{eq:Idmain} and specializing variables, and
\begin{equation}
\alpha(z)V(z)=-\frac12\big(V'(z)+\ftwo'(z)\big),
\end{equation}
which follows from \eqref{eq:alphatoV}--\eqref{eq:IdV}, into \eqref{eq:BTcalc2}. This yields
\begin{align}\label{eq:cjajddotcalc}
c_j\ddot{a}_j-2\dot{\lambda}=&\; -2\sum_{k\neq j}^N   \sum_{l\neq j,k}^N c_l\big(\ftwo'(a_j-a_k)-\ftwo'(a_k-a_l)\big)-4\sum_{k\neq j}^N c_jV'(a_j-a_k)-4\sum_{k\neq j}^N c_j \ftwo'(a_j-a_k).
\end{align}
Given $\sum_{l\neq j,k}^N c_l=-c_j-c_k$ by Lemma~\ref{lem:cjconstraint} and \eqref{eq:cjconstraint} and inserting \eqref{eq:lambdadot} into \eqref{eq:cjajddotcalc}, it follows that
\begin{align}\label{eq:cjajddot}
&c_j\ddot{a}_j-2\dot{\lambda} \nonumber \\
&=2\sum_{k\neq j}^N (c_j+c_k)\ftwo'(a_j-a_k)+2\sum_{k\neq j}^N\sum_{l\neq j,k}^N c_l\ftwo'(a_k-a_l)-4\sum_{k\neq j}^N c_jV'(a_j-a_k)-4\sum_{k\neq j}^N c_j \ftwo'(a_j-a_k)  \nonumber \\
&=-2\sum_{k\neq j}^N (c_j-c_k)\ftwo'(a_j-a_k)-\sum_{k\neq j}^N\sum_{l\neq j,k}^N (c_k-c_l)\ftwo'(a_k-a_l)-4\sum_{k\neq j}^N c_j V'(a_j-a_k),
\end{align}
where we have symmetrized the double sum, using that $\ftwo'(z)$ is an odd function \eqref{eq:parity}, in the second step. Employing the summation identity $\sum_{k=1}^N\sum_{l\neq k}^N s_{k,l}=\sum_{k\neq j}^N\sum_{l\neq j,k}^N s_{k,l}+2\sum_{k\neq j}^N s_{j,k}$ for symmetric complex arrays $\{s_{k,l}\}_{k,l=1}^N$ together with  \eqref{eq:lambdadot} in \eqref{eq:cjajddot} implies that 
\begin{equation}\label{eq:cjajddot2}
c_j\ddot{a}_j=-4\sum_{k\neq j}^N c_jV'(a_j-a_k).
\end{equation}
By Lemma~\ref{lem:cjneq0}, we may divide \eqref{eq:cjajddot2} through by $c_j$ to obtain \eqref{eq:CM}. This completes the proof.
\end{proof}

Consider the initial value problem consisting of the time evolution equations \eqref{eq:CM}, \eqref{eq:cjdot}, and \eqref{eq:lambdadot} and initial data satisfying \eqref{eq:ajdot} and \eqref{eq:cjconstraint} at $t=0$. By the Picard-Lindel\"{o}f theorem, this initial value problem admits a unique local solution, which may be uniquely extended (see, e.g., \cite[Corollary~3.2]{hartman1982}) as long as (i) no solution variable goes to infinity and (ii) \eqref{eq:ajakelliptic} holds (so that the functions defining the ODEs remain locally Lipschitz).

Similarly, consider the initial value problem consisting of the time evolution equations \eqref{eq:cjdot}--\eqref{eq:ajdot} and \eqref{eq:lambdadot} and the same initial data as imposed for the previous initial value problem (excepting the initial velocities of $\{a_j\}_{j=1}^N$; because \eqref{eq:ajdot} is a first-order equation, the initial velocities are not imposed, but will nonetheless match as \eqref{eq:ajdot} is a condition on the initial data in the first problem and a time evolution equation in the second). This initial value problem has a unique local solution, which may be uniquely extended as long as (i) no solution variable goes to infinity and (ii) \eqref{eq:ajakelliptic} and \eqref{eq:cjneq0} hold. However, Lemma~\ref{lem:cjneq0} shows that the condition \eqref{eq:cjneq0} cannot be violated. 

By Lemma~\ref{lem:CM}, the solutions to the two initial value problems must coincide on any interval $[0,\tau)$ on which \eqref{eq:ajakelliptic} holds; this implies the result.

\section{Details on the local limit}
\label{app:local}

The operators $T$ and $\tilde{T}$ in \eqref{eq:TT} satisfy \cite{scoufis2005,berntson2020}
\begin{equation}\label{eq:TTlimit}
\begin{split}
(Tf)(x)=&\; -\frac{1}{2\delta}\int_{-\infty}^x f(x')\,\mathrm{d}x'+\frac{1}{2\delta}\int_{x}^{\infty}f(x')\,\mathrm{d}x'+O(\delta), \\
(\tilde{T}f)(x)=&\; -\frac{1}{2\delta}\int_{-\infty}^x f(x')\,\mathrm{d}x'+\frac{1}{2\delta}\int_{x}^{\infty}f(x')\,\mathrm{d}x'+O(\delta)
\end{split}\quad \text{ as } \delta\downarrow 0.
\end{equation}

By inserting the expansions \eqref{eq:TTlimit} with $f=(|u|^2)_x$ and $f=(|v|^2)_x$ into \eqref{eq:IMM}, we obtain
\begin{equation}\label{eq:IMMlimit}
\begin{split}
\ii u_t=&\; u_{xx}+\ii u(|u|^2)_x-\frac{1}{\delta}u|u|^2+\frac{1}{2\delta}u(|u_{+\infty}|^2+|u_{-\infty}|^2)+\frac{1}{\delta}u|v|^2-\frac{1}{2\delta}u(|v_{+\infty}|^2+|v_{-\infty}|^2)+O(\delta), \\
\ii v_t=&\; v_{xx}+\ii v(|v|^2)_x+\frac{1}{\delta}v|v|^2-\frac{1}{2\delta}v(|v_{+\infty}|^2+|v_{-\infty}|^2)-\frac{1}{\delta}v|u|^2+\frac{1}{2\delta}v(|u_{+\infty}|^2+|u_{-\infty}|^2)+O(\delta),
\end{split}
\end{equation}
where
\begin{equation}
u_{\pm\infty}\coloneqq \lim_{x\to\pm\infty} u(x),\qquad v_{\pm\infty}\coloneqq \lim_{x\to\pm\infty} v(x).
\end{equation}
By imposing $|u_{+\infty}|^2+|u_{-\infty}|^2=|v_{+\infty}|^2+|v_{-\infty}|^2$ and rescaling $u\to \sqrt{\delta} u$ and $v\to \sqrt{\delta} v$ in \eqref{eq:IMMlimit}, we obtain \eqref{eq:manakov} with $\sigma_1=-\sigma_2=-1$ in the limit $\delta\downarrow 0$.

\section{Amplitudes of solutions}\label{app:amplitudes}

The derivation of \eqref{eq:moduli}--\eqref{eq:moduli2} and \eqref{eq:moduliperiodic}--\eqref{eq:B(t)} is provided in Appendix~\ref{app:moduli}. A proof of the constancy of the quantity $B$ in \eqref{eq:B(t)} is given in Appendix~\ref{app:B}.

\subsection{Derivation}\label{app:moduli}

Using \eqref{eq:ansatz2}, we write
\begin{equation}\label{eq:UUs1}
U\circ U^*=|\lambda|^2E-\ii\lambda \sum_{j=1}^N c_j^*A_-(x-a_j^*)+\ii\lambda^*\sum_{j=1}^N c_jA_+(x-a_j)+\sum_{j=1}^N\sum_{k=1}^N c_jc_k^*A_+(x-a_j)\circ A_-(x-a_k^*).
\end{equation}
We now insert \eqref{eq:solId1}, which gives
\begin{align}\label{eq:UUs2}
U\circ U^*=&\;|\lambda|^2E-\ii\lambda \sum_{k=1}^N c_k^*A_-(x-a_k^*)+\ii\lambda^*\sum_{j=1}^N c_jA_+(x-a_j) \nonumber \\
&\; +\sum_{j=1}^N\sum_{k=1}^N c_jc_k^*\alpha(a_j-a_k^*+\ii\delta)\big(A_+(x-a_j)-A_-(x-a_k^*)\big) \nonumber \\
&\; +\frac12\sum_{j=1}^N\sum_{k=1}^N c_jc_k^* \bigg(K_+(x-a_j)+K_-(x-a_k^*)+\ftwo(a_j-a_k^*+\ii\delta)E+\frac{3\zeta(\ii\delta;\ell,\ii\delta)}{\ii\delta} E\bigg) \nonumber \\
=&\; |\lambda|^2+\ii \sum_{j=1}^N c_j \Bigg( \lambda^*-\ii \sum_{k=1}^N c_k^*\alpha(a_j-a_k^*+\ii\delta)\Bigg)A_+(x-a_j) \nonumber \\
&\; -\ii \sum_{j=1}^N c_j^* \Bigg( \lambda+\ii \sum_{k=1}^N c_k\alpha(a_j^*-a_k-\ii\delta)\Bigg)A_-(x-a_j^*)  \nonumber \\
&\; +\frac12\sum_{j=1}^N\sum_{k=1}^N c_jc_k^* \bigg(K_+(x-a_j)+K_-(x-a_k^*)+\ftwo(a_j-a_k^*+\ii\delta)E+\frac{3\zeta(\ii\delta;\ell,\ii\delta)}{\ii\delta} E\bigg).
\end{align}

By using \eqref{eq:constraintperiodic}, we arrive at
\begin{align}\label{eq:UUs3}
U\circ U^*=&\; |\lambda|^2 E-\ii \sum_{j=1}^N \big(A_+(x-a_j)-A_-(x-a_j^*)\big) \nonumber \\
&\; +\frac12\sum_{j=1}^N\sum_{k=1}^N c_jc_k^* \bigg(K_+(x-a_j)+K_-(x-a_k^*)+\ftwo(a_j-a_k^*+\ii\delta)E+\frac{3\zeta(\ii\delta;\ell,\ii\delta)}{\ii\delta} E\bigg).
\end{align}	

We now consider the periodic and real-line cases separately.

\paragraph{Periodic case.}
By imposing the constraint \eqref{eq:cjconstraint}, the second line of \eqref{eq:UUs3} simplifies considerably,  
\begin{equation}
U\circ U^*=|\lambda|^2 E-\ii\sum_{j=1}^N \big(A_+(x-a_j)-A_-(x-a_j^*)\big)+\frac12\sum_{j=1}^N\sum_{k=1}^N c_jc_k^*\ftwo(a_j-a_k^*+\ii\delta)E.	
\end{equation}
By recalling the notation \eqref{eq:circ} and the definitions of $A_{\pm}(z)$ \eqref{eq:Apm} and $E$ \eqref{eq:AE}, we obtain \eqref{eq:moduliperiodic}--\eqref{eq:B(t)}.

\paragraph{Real-line case.}

We take the limit of \eqref{eq:UUs3} as $\ell\to \infty$. By using \eqref{eq:zetalimit2} and \eqref{eq:ftwolimit}, which together with the definition of \eqref{eq:Kpm} further imply 
\begin{equation}
\lim_{\ell\to\infty} K_{\pm}(z)=\bigg(\frac{\pi}{2\delta}\bigg)^2 E,	
\end{equation}
in \eqref{eq:UUs3}, we obtain
\begin{equation}\label{eq:UUsrealline}
U\circ U^*=\lambda^2 E-\ii\sum_{j=1}^N \big(A_+(x-a_j)-A_-(x-a_j^*)\big)+\bigg(\frac{\pi}{2\delta}\bigg)^2\sum_{j=1}^N\sum_{k=1}^N c_jc_k^* E,	
\end{equation}
where we have used that $\lambda\in \R$ in this case. By writing the last term in \eqref{eq:UUsrealline} as $\big(\frac{\pi}{2\delta}\big)^2 \big|\sum_{j=1}^N c_j\big|^2 E$ and recalling the notation \eqref{eq:circ} and the definitions of $A_{\pm}(z)$ \eqref{eq:Apm} and $E$ \eqref{eq:AE}, we obtain \eqref{eq:moduli}--\eqref{eq:moduli2}.

\subsection{Conservation of $B$}\label{app:B}

We prove that $B$ defined in \eqref{eq:B(t)} is conserved when the conditions of Theorem~\ref{thm:ellipticsolitons} are met. The precise statement is given and followed by the proof. 

\begin{proposition}\label{prop:B}
Let $u$ and $v$ be a solution of the periodic IMM system constructed in Theorem~\ref{thm:ellipticsolitons} on $[0,\tau)$ with corresponding parameters $\lambda$ and $\{a_j,c_j\}_{j=1}^N$. Then, $B$ defined in \eqref{eq:B(t)} is conserved on $[0,\tau)$.   	
\end{proposition}

\subsubsection{Proof of Proposition~\ref{prop:B}}

By integrating the first row of \eqref{eq:moduliperiodic} over $x\in [-\ell,\ell)$ and rearranging the result, we have
\begin{equation}\label{eq:2lB}
2\ell B=\int_{-\ell}^{\ell} |u|^2\,\mathrm{d}x +\ii \sum_{j=1}^N \int_{-\ell}^{\ell} \big(\alpha(x-a_j-\ii\delta/2)-\alpha(x-a_j^*+\ii\delta/2)\big)\,\mathrm{d}x.	
\end{equation}
We will show, in turn, that the first\footnote{We note that $\int_{-\ell}^{\ell}|v|^2\,\mathrm{d}x$ can similarly be shown to be conserved in time.} and second terms on the right-hand side of \eqref{eq:2lB} are conserved in time.

\begin{lemma}\label{lem:u2}
The quantity
\begin{equation}\label{eq:conserveduv}
\int_{-\ell}^{\ell} |u|^2\,\mathrm{d}x
\end{equation}	
is conserved when $u$ solves the periodic IMM system.
\end{lemma}

\begin{proof}
By differentiating the first quantity in \eqref{eq:conserveduv} with respect to $t$ and inserting \eqref{eq:IMM}, we obtain
\begin{align}\label{eq:uvconserved1}
&\frac{\mathrm{d}}{\mathrm{d}t}\int_{-\ell}^{\ell} |u|^2\,\mathrm{d}x	= \int_{-\ell}^{\ell} (u_tu^*+uu_t^*)\,\mathrm{d}x= \nonumber\\
 &\int_{-\ell}^{\ell} \big( \big(-\ii u_{xx}-\ii u(\ii+T)(|u|^2)_x+\ii u\tilde{T}(|v|^2)_x\big)u^*+u\big(\ii u_{xx}^*+\ii u^*(-\ii+T)(|u|^2)_x-\ii u^*\tilde{T}(|u|^2)_x\big)\big)\,\mathrm{d}x. 
\end{align}
After cancelling terms in \eqref{eq:uvconserved1}, we are left with
\begin{align}\label{eq:uvconserved2}
\frac{\mathrm{d}}{\mathrm{d}t}\int_{-\ell}^{\ell} |u|^2\,\mathrm{d}x=\int_{-\ell}^{\ell}\big(-\ii u_{xx}u^*+\ii u u_{xx}^*+2|u|^2(|u|^2)_x\big)\,\mathrm{d}x.
\end{align}
The integral of the first two terms vanishes by the self-adjointness of second derivative while the third term is a total derivative and integrates to zero. Hence $\int_{-\ell}^{\ell}|u|^2\,\mathrm{d}x$ is conserved on solutions of the periodic IMM system.
\end{proof}

To show that the second term in \eqref{eq:2lB} is conserved in time, we first note the following identity relating the $\zeta_1$- and $\zeta_2$-functions  (in what follows, we suppress the second and third arguments of these functions for notational simplicity) defined in \eqref{eq:zeta1} and \eqref{eq:zeta2}, respectively, 
\begin{equation}\label{eq:zeta2tozeta1}
\zeta_2(z)=\zeta_1(z)+\gamma_0 z	.
\end{equation}
Equation \eqref{eq:zeta2tozeta1} may be established by using the definitions \eqref{eq:zeta1} and \eqref{eq:zeta2}, the standard elliptic identity \cite[Eq.~23.2.14]{DLMF} $\zeta(\ell)(\ii\delta)-\zeta(\ii\delta)\ell=\ii\pi/2$ and the definition of $\gamma_0$ \eqref{eq:gamma0}. 

It thus follows from \eqref{eq:alphaelliptic} and \eqref{eq:zeta2tozeta1} that
\begin{multline}\label{eq:alphaint1}
	\int_{-\ell}^{\ell}	\big(\alpha(x-a_j-\ii\delta/2)-\alpha(x-a_j^*+\ii\delta/2)\big)\,\mathrm{d}x=\\ \int_{-\ell}^{\ell} \big(\zeta_1(x-a_j-\ii\delta/2)-\zeta_1(x-a_j^*+\ii\delta/2)\big)\,\mathrm{d}x	-\gamma_0\int_{-\ell}^{\ell}(a_j-a_j^*+\ii\delta)\,\mathrm{d}x.
\end{multline}
Recalling \eqref{eq:imaj}, we use the following exact integrals \cite[Proposition~B.1]{berntsonlangmann2021},
\begin{equation}
\int_{-\ell}^{\ell} \zeta_1(x-a)\,\mathrm{d}x=\begin{cases}
	+\ii \pi & -2\delta<\im(a)<0\\
	-\ii\pi & 0<\im(a)<2\delta 
\end{cases}	
\end{equation}
in \eqref{eq:alphaint1} to write
\begin{align}
	\label{eq:alphaint2}
	\ii\sum_{j=1}^N\int_{-\ell}^{\ell}	\big(\alpha(x-a_j-\ii\delta/2)-\alpha(x-a_j^*+\ii\delta/2)\big)\,\mathrm{d}x=&\;  -2 N\pi - 2\ii  \gamma_0\ell\sum_{j=1}^N(a_j-a_j^*)+2N\gamma_0\ell\delta \nonumber\\=&\; -N\pi+\frac{2\pi}{\delta}\sum_{j=1}^N \im(a_j),
\end{align}
where we have inserted the definition of $\gamma_0$ \eqref{eq:gamma0} in the second step. 

It remains to show that \eqref{eq:alphaint2} is conserved in time; this is accomplished with the following lemma.

\begin{lemma}\label{lem:ajsum}
Let $\lambda$ and $\{a_j,c_j\}$ be a solution of \eqref{eq:CM}, \eqref{eq:cjdot}, and \eqref{eq:lambdadot} on $[0,\tau)$ such that \eqref{eq:ajdot}, \eqref{eq:cjconstraint}, and \eqref{eq:constraintperiodic} hold at $t=0$ and \eqref{eq:imaj} and \eqref{eq:ajakelliptic} hold on $[0,\tau)$. Then, the sum of the imaginary parts of $\{a_j\}_{j=1}^N$,
\begin{equation}\label{eq:imajsum}
	\sum_{j=1}^N \im(a_j)=\frac1{2\ii}\sum_{j=1}^N (a_j-a_j^*)
\end{equation}	
is conserved on $[0,\tau)$. 
\end{lemma}

\begin{proof}
By Lemma~\ref{lem:cjconstraint}, Proposition~\ref{prop:constraints}, and Proposition~\ref{prop:CM}, we have that \eqref{eq:cjconstraint}, \eqref{eq:constraintperiodic}, and \eqref{eq:ajdot}, respectively, hold on $[0,\tau)$ and thus may be used freely in what follows. Also, we repeatedly use the fact that $\alpha(z)$ is a $2\ii\delta$-periodic \eqref{eq:imperiod} and odd \eqref{eq:parity} function below. 

We differentiate \eqref{eq:imajsum} with respect to $t$ and insert \eqref{eq:ajdot} to obtain  
\begin{align}\label{eq:imajconserved1}
\frac{\mathrm{d}}{\mathrm{d}t}\sum_{j=1}^N \im(a_j)=&\; \frac{1}{2\ii}\sum_{j=1}^N (\dot{a}_j-\dot{a}_j^*) \nonumber\\
=&\; -\ii\sum_{j=1}^N \frac{1}{c_j}\Bigg(\lambda+\ii\sum_{k\neq j}^N c_k\alpha(a_j-a_k)\Bigg)+\ii\sum_{j=1}^N \frac{1}{c_j^*}\Bigg(\lambda^*-\ii\sum_{k\neq j}^N c_k^*\alpha(a_j^*-a_k^*)\Bigg).
\end{align}
To proceed, we solve \eqref{eq:constraintperiodic} for $1/c_j$ and substitute the result into \eqref{eq:imajconserved1}; this gives
\begin{align}\label{eq:imajconserved2}
\frac{\mathrm{d}}{\mathrm{d}t}\sum_{j=1}^N \im(a_j)=&\; \ii\sum_{j=1}^N \Bigg(\lambda^*-\ii\sum_{l=1}^N c_l^*\alpha(a_j-a_l^*+\ii\delta)\Bigg)\Bigg(\lambda+\ii\sum_{k\neq j}^N c_k\alpha(a_j-a_k)\Bigg) \nonumber \\
&\; -\ii\sum_{j=1}^N \Bigg(\lambda+\ii\sum_{l=1}^N c_l\alpha(a_j^*-a_l+\ii\delta)\Bigg)\Bigg(\lambda^*-\ii\sum_{k\neq j}^N c_k^*\alpha(a_j^*-a_k^*)\Bigg) \nonumber \\
=&\; - \lambda^*\sum_{j=1}^N\sum_{k\neq j}^N c_k\alpha(a_j-a_k)-\lambda\sum_{j=1}^N\sum_{k\neq j}^N c_k^*\alpha(a_j^*-a_k^*) \nonumber\\
&\; +\lambda \sum_{j=1}^N\sum_{l=1}^N c_l^*\alpha(a_j-a_l^*+\ii\delta)+\lambda^*\sum_{j=1}^N\sum_{l=1}^N c_l\alpha(a_j^*-a_l+\ii\delta) \nonumber \\
&\; +\ii\sum_{j=1}^N\sum_{k\neq j}^N \sum_{l=1}^N \big(c_kc_l^*\alpha(a_j-a_k)\alpha(a_j-a_l^*+\ii\delta)-c_k^*c_l\alpha(a_j^*-a_k^*)\alpha(a_j^*-a_l+\ii\delta)      \big).
\end{align}

Next, we use \eqref{eq:constraintperiodic} to replace the quantities $\lambda^* c_k$ and $\lambda^*c_l$ (and their complex conjugates)
in \eqref{eq:imajconserved2}:
\begin{align}\label{eq:imajconserved3}
\frac{\mathrm{d}}{\mathrm{d}t}\sum_{j=1}^N \im(a_j)=&\; 	\sum_{j=1}^N\sum_{k\neq j}^N \alpha(a_j-a_k)-\ii \sum_{j=1}^N\sum_{k\neq j}^N\sum_{l=1}^N c_kc_l^*\alpha(a_j-a_k)\alpha(a_k-a_l^*+\ii\delta) \nonumber \\
&\; +\sum_{j=1}^N\sum_{k\neq j}^N \alpha(a_j^*-a_k^*)+\ii\sum_{j=1}^N\sum_{k\neq j}^N\sum_{l=1}^N c_k^*c_l\alpha(a_j^*-a_k^*)\alpha(a_k^*-a_l+\ii\delta) \nonumber \\
&\; -\sum_{j=1}^N\sum_{l=1}^N \alpha(a_j-a_l^*+\ii\delta)+\ii\sum_{j=1}^N\sum_{k=1}^N\sum_{l=1}^N c_kc_l^* \alpha(a_j-a_l^*+\ii\delta)\alpha(a_k-a_l^*+\ii\delta) \nonumber \\
&\; -\sum_{j=1}^N\sum_{l=1}^N \alpha(a_j^*-a_l+\ii\delta)-\ii\sum_{j=1}^N\sum_{k=1}^N\sum_{l=1}^N c_k^*c_l\alpha(a_j^*-a_l+\ii\delta)\alpha(a_k^*-a_l+\ii\delta) \nonumber \\
&\; +\ii\sum_{j=1}^N\sum_{k\neq j}^N \sum_{l=1}^N \big(c_kc_l^*\alpha(a_j-a_k)\alpha(a_j-a_l^*+\ii\delta)-c_k^*c_l\alpha(a_j^*-a_k^*)\alpha(a_j^*-a_l+\ii\delta)      \big).
\end{align}
All double sums in \eqref{eq:imajconserved3} vanish by symmetry; note also that the $k=j$ terms in the triple sums in the third and fourth lines cancel with each other. With these simplifications, \eqref{eq:imajconserved3} can be written as
\begin{align}\label{eq:imajconserved4}
&\frac{\mathrm{d}}{\mathrm{d}t}\sum_{j=1}^N \im(a_j)= \nonumber\\
& \ii\sum_{j=1}^N\sum_{k\neq j}^N\sum_{l=1}^N c_kc_l^* \big(\alpha(a_j-a_l^*+\ii\delta)\alpha(a_k-a_l^*+\ii\delta)- \alpha(a_j-a_k)\alpha(a_k-a_l^*+\ii\delta)+\alpha(a_j-a_k)\alpha(a_j-a_l^*+\ii\delta)\big) \nonumber\\
& -\ii\sum_{j=1}^N\sum_{k\neq j}^N\sum_{l=1}^N c_k^*c_l \big(\alpha(a_j^*-a_l+\ii\delta)\alpha(a_k^*-a_l+\ii\delta)- \alpha(a_j^*-a_k^*)\alpha(a_k^*-a_l+\ii\delta)+\alpha(a_j^*-a_k^*)\alpha(a_j^*-a_l+\ii\delta)\big).
\end{align}

Note that the first and second lines in \eqref{eq:imajconserved4} are related by complex conjugation. Any manipulation of the first line can be extended to the second line using this symmetry. We apply \eqref{eq:Idmain} with $z=a_j+\ii\delta/2$, $a=a_k+\ii\delta/2$, and $b=a_l^*-\ii\delta/2$ to \eqref{eq:imajconserved4}, which leads to 
\begin{align}\label{eq:imajconserved5}
\frac{\mathrm{d}}{\mathrm{d}t}\sum_{j=1}^N \im(a_j)= &\; \frac{\ii}{2}\sum_{j=1}^N\sum_{k\neq j}^N\sum_{l=1}^N c_kc_l^* \bigg( \ftwo(a_j-a_k)+\ftwo(a_j-a_l^*+\ii\delta)+\ftwo(a_k-a_l^*+\ii\delta) +\frac{3\zeta(\ii\delta)}{\ii\delta}    \bigg) \nonumber\\
&\; -\frac{\ii}{2}\sum_{j=1}^N\sum_{k\neq j}^N\sum_{l=1}^N c_k^*c_l \bigg( \ftwo(a_j^*-a_k^*)+\ftwo(a_j^*-a_l+\ii\delta)+\ftwo(a_k^*-a_l+\ii\delta) +\frac{3\zeta(\ii\delta)}{\ii\delta}    \bigg) .
\end{align}
By using \eqref{eq:cjconstraint} (which implies $\sum_{k\neq j}^N c_k=-c_j$) in \eqref{eq:imajconserved5}, we arrive at  
\begin{align}\label{eq:imajconserved6}
\frac{\mathrm{d}}{\mathrm{d}t}\sum_{j=1}^N \im(a_j)= &\;-\frac{\ii}{2}\sum_{j=1}^N\sum_{l=1}^N\big(c_jc_l^*\ftwo(a_j-a_l^*+\ii\delta)-c_j^*c_l\ftwo(a_j^*-a_l+\ii\delta)\big)
\nonumber \\
&\; +\frac{\ii}{2}\sum_{j=1}^N\sum_{k\neq j}^N\sum_{l=1}^N \big(c_kc_l^*\ftwo(a_k-a_l^*+\ii\delta)	-c_k^*c_l\ftwo(a_k^*-a_l+\ii\delta)\big). 
\end{align}
Both sums in \eqref{eq:imajconserved6} vanish by symmetry, using the fact that $\ftwo(z)$ is an even function \eqref{eq:parity}, and the result follows. 
\end{proof}

Applying Lemma~\ref{lem:u2} to our supposed solution of the periodic IMM system on $[0,\tau)$ and using Lemma~\ref{lem:ajsum} in \eqref{eq:alphaint2}, we find that the right-hand side of \eqref{eq:2lB} is conserved on $[0,\tau)$. It follows that $B$ is constant-in-time. 

\bibliographystyle{unsrt}

\bibliography{BF1}

\begin{thebibliography}{10}

\bibitem{benney1967}
D.J. Benney and A.C. Newell.
\newblock {The Propagation of Nonlinear Wave Envelopes}.
\newblock {\em Stud. Appl. Math.}, 46(1-4):133--139, 1967.

\bibitem{ablowitz2003discrete}
M.J. Ablowitz, B.~Prinari, and A.D. Trubatch.
\newblock {\em {Discrete and Continuous Nonlinear Schr\"{o}dinger Systems}}.
\newblock Cambridge University Press, Cambridge, 2003.

\bibitem{manakov1974}
S.V. Manakov.
\newblock On the theory of two-dimensional stationary self-focusing of
  electromagnetic waves.
\newblock {\em Sov. Phys. JETP}, 38(2):248--253, 1974.

\bibitem{zakharov1982}
V.E. Zakharov and E.I. Schulman.
\newblock {To the integrability of the system of two coupled nonlinear
  Schr\"{o}dinger equations}.
\newblock {\em Physica D}, 4(2):270--274, 1982.

\bibitem{berkhoer1970}
A.L. Berkhoer and V.E. Zakharov.
\newblock Self excitation of waves with different polarizations in nonlinear
  media.
\newblock {\em Sov. Phys. JETP}, 31(3):486--490, 1970.

\bibitem{baronio2012}
F.~Baronio, A.~Degasperis, M.~Conforti, and S.~Wabnitz.
\newblock {Solutions of the Vector Nonlinear Schr\"odinger Equations: Evidence
  for Deterministic Rogue Waves}.
\newblock {\em Phys. Rev. Lett.}, 109:044102, 2012.

\bibitem{bludov2010}
Y.V. Bludov, V.V. Konotop, and N.~Akhmediev.
\newblock {Vector rogue waves in binary mixtures of Bose-Einstein condensates}.
\newblock {\em Eur. Phys. J. Spec. Top.}, 185(1):169--180, 2010.

\bibitem{tian2016}
S.-F. Tian.
\newblock {The mixed coupled nonlinear Schr\"{o}dinger equation on the
  half-line via the Fokas method}.
\newblock {\em Proc. Royal Soc. A}, 472(2195):20160588, 2016.

\bibitem{tian2017}
S.-F. Tian.
\newblock {Initial--boundary value problems for the general coupled nonlinear
  Schr\"{o}dinger equation on the interval via the Fokas method}.
\newblock {\em J. Diff. Eq.}, 262(1):506--558, 2017.

\bibitem{kanna2006}
T.~Kanna, M.~Lakshmanan, P.~Tchofo Dinda, and N.~Akhmediev.
\newblock {Soliton collisions with shape change by intensity redistribution in
  mixed coupled nonlinear Schr\"odinger equations}.
\newblock {\em Phys. Rev. E}, 73:026604, 2006.

\bibitem{ohta2011general}
Y.~Ohta, D.-S. Wang, and J.~Yang.
\newblock {General $N$-Dark--Dark Solitons in the Coupled Nonlinear
  Schr\"{o}dinger Equations}.
\newblock {\em Stud. Appl. Math.}, 127(4):345--371, 2011.

\bibitem{feng2014}
B.-F. Feng.
\newblock {General $N$-soliton solution to a vector nonlinear Schr{\"o}dinger
  equation}.
\newblock {\em J. Phys. A: Math. Theor.}, 47(35):355203, 2014.

\bibitem{matsuno2000}
Y.~Matsuno.
\newblock {Multiperiodic and multisoliton solutions of a nonlocal nonlinear
  Schr{\"o}dinger equation for envelope waves}.
\newblock {\em Phys. Lett. A}, 278(1):53 -- 58, 2000.

\bibitem{gerard2022}
P.~G{\'e}rard and E.~Lenzmann.
\newblock {The Caloger--Moser Derivative Nonlinear Schr{\"o}dinger Equation}.
\newblock arXiv preprint: 2208.0415 (math.AP), 2022.

\bibitem{pelinovsky1995}
D.~Pelinovsky.
\newblock {Intermediate nonlinear Schr\"{o}dinger equation for internal waves
  in a fluid of finite depth}.
\newblock {\em Phys. Lett. A}, 197(5):401--406, 1995.

\bibitem{pelinovsky1995spectral}
D.E. Pelinovsky and R.H.J. Grimshaw.
\newblock {A spectral transform for the intermediate nonlinear Schr\"{o}dinger
  equation}.
\newblock {\em J. Math. Phys.}, 36(8):4203--4219, 1995.

\bibitem{matsuno2001}
Y.~Matsuno.
\newblock {$N$-soliton formulae for the intermediate nonlinear Schr\"{o}dinger
  equation}.
\newblock {\em Inverse Probl.}, 17:501--514, 2001.

\bibitem{tutiya2009}
Y.~Tutiya.
\newblock {Bright $N$-solitons for the intermediate nonlinear Schr\"{o}dinger
  equation}.
\newblock {\em J. Nonlinear Math. Phys.}, 16(1):7--23, 2009.

\bibitem{matsuno2002exactly}
Y.~Matsuno.
\newblock {Exactly solvable eigenvalue problems for a nonlocal nonlinear
  Schr\"{o}dinger equation}.
\newblock {\em Inverse Probl.}, 18:1101--1125, 2002.

\bibitem{matsuno2004}
Y.~Matsuno.
\newblock {A Cauchy problem for the nonlocal nonlinear Schr{\"o}dinger
  equation}.
\newblock {\em Inverse Probl.}, 20(2):437--445, 2004.

\bibitem{matsuno2002}
Y.~Matsuno.
\newblock {Calogero--Moser--Sutherland Dynamical Systems Associated with
  Nonlocal Nonlinear Schr{\"o}dinger Equation for Envelope Waves}.
\newblock {\em J. Phys. Soc. Japan}, 71(6):1415--1418, 2002.

\bibitem{berntson2020}
B.K. Berntson, E.~Langmann, and J.~Lenells.
\newblock {Non-chiral Intermediate Long Wave equation and inter-edge effects in
  narrow quantum Hall systems}.
\newblock {\em Phys. Rev. B}, 102:155308, 2020.

\bibitem{berntson2022}
B.K. Berntson, R.~Klabbers, and E.~Langmann.
\newblock {The non-chiral intermediate Heisenberg ferromagnet equation}.
\newblock {\em J. High Energ. Phys.}, 2022(3):46, 2022.

\bibitem{berntsonlangmannlenells2022}
B.K. Berntson, E.~Langmann, and J.~Lenells.
\newblock {Spin generalizations of the Benjamin-Ono equation}.
\newblock {\em Lett. Math. Phys.}, 112:50, 2022.

\bibitem{olshanetsky1983}
M.A. Olshanetsky and A.M. Perelomov.
\newblock Classical integrable finite-dimensional systems related to {Lie}
  algebras.
\newblock {\em Physics Reports}, 71(5):313--400, 1981.

\bibitem{krichever1980}
I.M. Krichever.
\newblock {Elliptic solutions of the Kadomtsev-Petviashvili equation and
  integrable systems of particles}.
\newblock {\em Funct. Anal. Appl.}, 14(4):282--290, 1980.

\bibitem{kruskal1974}
M.D. Kruskal.
\newblock {The Korteweg-de Vries equation and related evolution equations}.
\newblock {\em Lect. Appl. Math.}, 15:61--83, 1974.

\bibitem{thickstun1976}
W.R. Thickstun.
\newblock A system of particles equivalent to solitons.
\newblock {\em J. Math. Anal. Appl.}, 55(2):335--346, 1976.

\bibitem{airault1977}
H.~Airault, H.P. McKean, and J.~Moser.
\newblock {Rational and elliptic solutions of the Korteweg-de Vries equation
  and a related many-body problem}.
\newblock {\em Comm. Pure Appl. Math.}, 30(1):95--148, 1977.

\bibitem{choodnovsky1977}
D.V. Choodnovsky and G.V. Choodnovsky.
\newblock Pole expansions of nonlinear partial differential equations.
\newblock {\em Nuovo Cim. B}, 40(2):339--353, 1977.

\bibitem{ramani1981}
A.~Ramani.
\newblock {Inverse scattering, ordinary differential equations of Painlev\'{e}
  type, and Hirota's bilinear formalism}.
\newblock {\em Ann. N.Y. Acad. Sci.}, 373(1):54--67, 1981.

\bibitem{hietarinta1987}
J.~Hietarinta.
\newblock {A search for bilinear equations passing Hirota's three‐-soliton
  condition. I. KdV‐-type bilinear equations}.
\newblock {\em J. Math. Phys.}, 28(8):1732--1742, 1987.

\bibitem{berntsonlangmann2020}
B.K. Berntson, E.~Langmann, and J.~Lenells.
\newblock On the non-chiral intermediate long wave equation.
\newblock {\em Nonlinearity}, 35:4549--4584, 2022.

\bibitem{wojciechowski1982}
S.~Wojciechowski.
\newblock {The analogue of the B\"{a}cklund transformation for integrable
  many-body systems}.
\newblock {\em J. Phys. A: Math. Theor.}, 15(12):L653--L657, 1982.

\bibitem{wilson1998}
G.~Wilson.
\newblock {Collisions of Calogero-Moser particles and an adelic Grassmannian
  (With an Appendix by I.G. Macdonald)}.
\newblock {\em Invent. Math.}, 133(1):1--41, 1998.

\bibitem{berntsonlangmann2021}
B.K. Berntson, E.~Langmann, and J.~Lenells.
\newblock On the non-chiral intermediate long wave equation: {II}. periodic
  case.
\newblock {\em Nonlinearity}, 35:4517--4548, 2022.

\bibitem{berntson2022elliptic}
B.K. Berntson, E.~Langmann, and J.~Lenells.
\newblock {Elliptic soliton solutions of the spin non-chiral intermediate long
  wave equation}.
\newblock arXiv preprint: 2211.13791 (math-ph), 2022.

\bibitem{DLMF}
{\it NIST Digital Library of Mathematical Functions}.
\newblock http://dlmf.nist.gov/, Release 1.0.26 of 2020-03-15.
\newblock F.W.J. Olver, A.B. {Olde Daalhuis}, D.W. Lozier, B.I. Schneider, R.F.
  Boisvert, C.W. Clark, B.R. Miller, B.V. Saunders, H.S. Cohl, and M.A.
  McClain, eds.

\bibitem{manin1998}
Y.I. Manin.
\newblock {Sixth Painlev\'{e} equation, universal elliptic curve, and mirror of
  $\mathbb{P}^2$}.
\newblock {\em Am. Math. Soc. Trans.}, 186:131--151, 1998.

\bibitem{berntsonklabbers2020}
B.K. Berntson, R.~Klabbers, and E.~Langmann.
\newblock Multi-solitons of the half-wave maps equation and spin-pole
  {Calogero-Moser} dynamics.
\newblock {\em J. Phys. A: Math. Theor.}, 53:505702, 2020.

\bibitem{abanov2009}
A.G. Abanov, E.~Bettelheim, and P.~Wiegmann.
\newblock {Integrable hydrodynamics of Calogero-Sutherland model: bidirectional
  Benjamin-Ono equation}.
\newblock {\em J. Phys. A: Math. Theor.}, 42(13):135201, 2009.

\bibitem{lenzmann2020}
E.~Lenzmann and J.~Sok.
\newblock {Derivation of the Half-Wave Maps Equation from Calogero--Moser Spin
  Systems}.
\newblock arXiv preprint: 2007.15323 (math.AP), 2020.

\bibitem{ahrend2022}
M.~Ahrend and E.~Lenzmann.
\newblock {Uniqueness for the nonlocal Liouville equation in $\mathbb{R}$}.
\newblock {\em J. Funct. Anal.}, 283(12):109712, 2022.

\bibitem{calogero1976}
F.~Calogero.
\newblock {A sequence of Lax matrices for certain integrable Hamiltonian
  systems}.
\newblock {\em Lett. Nuovo Cim.}, 16:22--24, 1976.

\bibitem{hartman1982}
P.~Hartman.
\newblock {\em Ordinary differential equations}.
\newblock Birkh\"{a}user, Boston, Massachusetts, reprint of the second edition,
  1982.

\bibitem{scoufis2005}
G.~Scoufis and C.M. Cosgrove.
\newblock An application of the inverse scattering transform to the modified
  intermediate long wave equation.
\newblock {\em J. Math. Phys.}, 46(10):103501, 2005.

\end{thebibliography}

\end{document}